\documentclass{article}

\usepackage{xspace}

\def\showdraftbox{1}
\def\showauthornotes{1}

\usepackage{hyperref}
\usepackage{amstext,amsfonts,bbm,algorithm,algorithmicx,xspace,nicefrac}
\usepackage{color,stmaryrd,enumerate,latexsym,bm,amsfonts,
  wrapfig,verbatim,textcomp}
\usepackage[small]{caption}
\usepackage{comment} 
\usepackage{epsfig} 
\usepackage{tabularx}
\usepackage{tabu}
\usepackage{fancybox,graphicx,url}
\usepackage[inline]{enumitem}
\usepackage{fullpage}
\usepackage{booktabs}
\usepackage{commath}

\newcommand{\defeq}{\stackrel{\textup{def}}{=}}

\newtheorem{theorem}{Theorem}[section]
\newtheorem{lemma}[theorem]{Lemma}
\newtheorem{definition}[theorem]{Definition} 

\newtheorem{fact}[theorem]{Fact}
\newtheorem{remark}[theorem]{Remark}

\newcommand{\diag}[1]{{\bf Diag}\left({#1}\right)}

\newcommand{\nfrac}[2]{\nicefrac{#1}{#2}}
\def\abs#1{\left| #1 \right|}
\renewcommand{\norm}[1]{\ensuremath{\left\lVert #1 \right\rVert}}

\newcommand{\ceil}[1]{\left\lceil\, {#1}\,\right\rceil}

\newcommand\rea{\mathbb R}

\newcommand{\marginlabel}[1]%
{\mbox{}\marginpar{\it{\raggedleft\hspace{0pt}#1}}}
\newcommand{\poly}{\mathrm{poly}}
\newcommand{\polylog}{{\mathrm{\mbox{polylog}}}}

\newcommand\calC{\mathcal{C}}

\definecolor{Mygray}{gray}{0.8}

 \ifcsname ifcommentflag\endcsname\else
  \expandafter\let\csname ifcommentflag\expandafter\endcsname
                  \csname iffalse\endcsname
\fi

\ifnum\showauthornotes=1

\else

\fi

\ifnum\showauthornotes=1
\newcommand{\Authornote}[2]{{\sf\color{red}{[#1: #2]}}}
\newcommand{\Authoredit}[2]{{\sf\color{red}{[#1]}\color{blue}{#2}}}
\newcommand{\Authorcomment}[2]{{\sf \color{gray}{[#1: #2]}}}
\newcommand{\Authorfnote}[2]{\footnote{\color{red}{#1: #2}}}
\newcommand{\Authorfixme}[1]{\Authornote{#1}{\textbf{??}}}
\newcommand{\Authormarginmark}[1]{\marginpar{\textcolor{red}{\fbox{
#1:!}}}}
\else
\newcommand{\Authornote}[2]{}
\newcommand{\Authoredit}[2]{}
\newcommand{\Authorcomment}[2]{}
\newcommand{\Authorfnote}[2]{}
\newcommand{\Authorfixme}[1]{}
\newcommand{\Authormarginmark}[1]{}
\fi

\newcommand{\paren}[1]{\left({#1}\right)}

\newenvironment{fminipage}%
  {\begin{Sbox}\begin{minipage}}%
  {\end{minipage}\end{Sbox}\fbox{\TheSbox}}

\newlength{\pgmtab}  
\newcommand {\ELSE}{{\bf else\ }}
\newcommand {\IF}{{\bf if\ }}
\newcommand {\FOR}{{\bf for\ }}

\newcommand {\RETURN}{\mbox{\bf return\ }}

 {
	\begin{enumerate}}{\end{enumerate}}

\def\qedsketch{\ifmmode\Box\else{\unskip\nobreak\hfil
\penalty50\hskip1em\null\nobreak\hfil$\Box$
\parfillskip=0pt\finalhyphendemerits=0\endgraf}\fi}

\newenvironment{proof}{\begin{trivlist} \item {\bf Proof:~~}}
   {\hfill $\Box$ \end{trivlist}}

\newenvironment{proofof}[1]{\begin{trivlist} \item {\bf Proof
#1:~~}}
  { \hfill $\Box$ \end{trivlist}}

\newlength{\tpush}
\newcommand{\handout}[5]{
   \noindent
   \begin{center}
   \framebox{ \vbox{ \hbox to \textwidth { {\bf \coursenum\ :\  \coursename} \hfill #5 }
       \vspace{3mm}
       \hbox to \textwidth { {\Large \hfill #2  \hfill} }
       \vspace{1mm}
       \hbox to \textwidth { {\it #3 \hfill #4} }
     }
   }
   \end{center}
   \vspace*{4mm}
   \newcommand{\lecturenum}{#1}
   \addcontentsline{toc}{chapter}{Lecture #1 -- #2}
}

\ifnum\showdraftbox=1

\else

\fi

\usepackage{alltt} 

\newtheorem{claim}[theorem]{Claim}

\def\prob#1#2{{\mathbb{P}}_{#1}\left[ #2 \right]}

\def\defeq{\stackrel{\mathrm{def}}{=}}
\def\setof#1{\left\{#1  \right\}}

\def\ceil#1{\left\lceil #1 \right\rceil}

\DeclareMathOperator*{\trop}{Tr}

\def\abs#1{\left|#1  \right|}

\def\norm#1{\left\| #1 \right\|}

\def\calC{\mathcal{C}}

\def\aa{\pmb{\mathit{a}}}
\newcommand\bb{\boldsymbol{\mathit{b}}}
\newcommand\cc{\boldsymbol{\mathit{c}}}

\newcommand\ee{\boldsymbol{\mathit{e}}}

\newcommand{\AND}{\quad \text{and} \quad}

\usepackage{verbatim}

\usepackage{cases}
\usepackage{bbm}

\title{
Approximate
Gaussian Elimination for Laplacians \\
 -- Fast,
  Sparse, and Simple}

\author{
  Rasmus Kyng\thanks{Supported by 
   NSF grant CCF-1111257 and ONR Award N00014-16-1-2374.} \\
Yale University\\
rasmus.kyng@yale.edu
  \and
  Sushant Sachdeva\thanks{Supported by a Simons Investigator Award to Daniel A. Spielman.}\\
Yale University\\
sachdeva@cs.yale.edu
}
 
\allowdisplaybreaks[4]

\newcommand{\eps}{\epsilon}
\newcommand{\lapid}{\Pi}

\DeclareMathOperator*{\expt}{\mathbb{E}}

\newcommand{\curlbr}[1]{\left\{#1  \right\}}
\newcommand{\sqbr}[1]{\left[#1  \right]}

\newcommand{\vecone}{\bm{1}}
\newcommand{\veczero}{\bm{0}}
\newcommand{\matzero}{\bm{0}}

\newcommand{\ind}[1]{\mathbbm{1}_{#1}}
\newcommand{\trunc}[1]{\widetilde{#1}}
\newcommand{\samp}[1]{\widehat{#1}}
\newcommand{\lnorm}[1]{\ensuremath{\overline{#1}}}
\newcommand{\vstar}[2]{\ensuremath{\left(#1\right)_{#2}}}

\newcommand{\mult}[2]{\ensuremath{\deg_{#1}(#2)}}

\newcommand{\seqexpt}[1]{\ensuremath{ 
\expt_{(#1)}}}

\newcommand{\truncSum}{truncated martingale}

\newcommand{\truncZ}{\trunc{Z}}

\newcommand{\seqsum}{bags-of-dice martingale}

\newcommand{\SEQSUM}{Bags-of-Dice Martingale}

\newcommand{\OR}{\text{ or }}
\newcommand{\goodUntil}{\mathcal{A}}

\newcommand{\assign}{\leftarrow}

\newcommand{\csamp}{\textsc{CliqueSample}}
\newcommand{\cholesky}{\textsc{SparseCholesky}}
\newcommand{\concChol}{\textsc{LowDegreeSparseCholesky}}

\usepackage{caption,mdframed,setspace}

\newlength{\algtopspace}
\newlength{\algpostcaptionspace}
\newenvironment{tight_enumerate}{
\begin{enumerate}[itemsep=1.5pt,parsep=0pt,leftmargin=15pt,rightmargin=0pt,topsep=3pt]
\setlength{\parskip}{0pt}
}{\end{enumerate}}

\makeatletter

\newenvironment{shiftedflalign*}{%
    \start@align\tw@\st@rredtrue\m@ne
    \hspace{1.5 em}
}{%
    \endalign
}

\begin{document}

\maketitle
\thispagestyle{empty}

\begin{abstract}

  We show how to perform sparse approximate Gaussian elimination for
  Laplacian matrices.  We present a simple, nearly linear time
  algorithm that approximates a Laplacian by a matrix with a sparse
  Cholesky factorization -- the version of Gaussian elimination for
  symmetric matrices.
  This is the first nearly linear time solver for Laplacian systems
  that is based purely on random sampling, and does not use any graph
  theoretic constructions such as low-stretch trees, sparsifiers, or
  expanders. The crux of our analysis 
is a novel
  concentration bound for matrix martingales where the differences
  are sums of conditionally independent variables.
\end{abstract}

\clearpage
\setcounter{page}{1}

\section{Introduction}
A symmetric matrix $L$ is called Symmetric and Diagonally Dominant
(SDD) if for all $i,$ $L(i,i) \ge \sum_{j \neq i} |L(i,j)|.$
An SDD matrix $L$ is a Laplacian if $L(i,j) \le 0$ for $i \neq j,$ and 
for all $i,$ $\sum_{j} L(i,j) = 0.$ A Laplacian matrix is
naturally associated with a graph on its vertices, where $i,j$ are
adjacent if $L(i,j) \neq 0.$
The problem of solving systems of linear equations $Lx=b,$ where $L$
is an SDD matrix (and often a Laplacian), is a fundamental primitive
and arises in varied applications in both theory and practice. Example
applications include solutions of partial differential equations via
the finite element method~\cite{Strang86, BomanHV04}, semi-supervised
learning on graphs~\cite{Zhu03, ZhouS04, ZhouBLWS04}, and computing
maximum flows in graphs~\cite{DaitchS08, ChristianoKMST10, Madry13,
  LeeS13}. It has also been used as a primitive in the design of several
fast algorithms~\cite{KelnerM09, OrecchiaSV12, KelnerMillerPeng,
  LevinKoutisPeng, KyngRS15}. It is known that solving SDD
linear systems can be reduced to solving Laplacian
systems~\cite{Gremban96}.

\paragraph{Cholesky Factorization.}
A natural approach to solving systems of linear equations is Gaussian
elimination, or its variant for symmetric matrices, Cholesky
factorization. Cholesky factorization of a matrix $L$
produces a factorization 
$L = \mathcal{L} \mathcal{D} \mathcal{L}^{\top},$ where $\mathcal{L}$
is a lower-triangular matrix, and $\mathcal{D}$ is a diagonal
matrix. Such a factorization allows us to solve a system $Lx = b$ by
computing $x = L^{-1} b = (\mathcal{L}^{-1})^{\top} \mathcal{D}^{-1}
\mathcal{L}^{-1} b,$ where the inverse of $\mathcal{L}$, and
$\mathcal{D}$ 
can be applied quickly since they are
lower-triangular and diagonal,
respectively.

The fundamental obstacle to using Cholesky factorization for quickly
solving systems of linear equations is that $\mathcal{L}$ can be a
dense matrix even if the original matrix $L$ is sparse. The reason is
that the key step in Cholesky factorization, eliminating a variable,
say $x_{i},$ from a system of equations, creates a new coefficient
$L^{\prime}({j,k})$ for every pair $j,k$ such that $L({j,i})$ and
$L({i,k})$ are non-zero. This phenomenon is called \emph{fill-in}.
For Laplacian systems, eliminating the first variable corresponds to
eliminating the first vertex in the graph, and the fill-in corresponds
to adding a clique on all the neighbors of the first
vertex. Sequentially eliminating variables often produces a sequence
of increasingly-dense systems, resulting in an $O(n^{3})$ worst-case
time even for sparse $L.$
Informally, the algorithm for generating the Cholesky factorization
for a Laplacian can be expressed as follows:
\begin{enumerate}[label={},itemsep=1pt,parsep=0pt,leftmargin=10pt,rightmargin=20pt,topsep=5pt]
\item[] \texttt{\FOR $i = 1$ to $n-1$
  \item[] \hspace{\pgmtab} Use equation $i$ to express
    the variable for vertex $i$ in terms of the remaining variables.  
  \item[] \hspace{\pgmtab} Eliminate vertex $i$, adding a clique on
    the neighbors of $i.$}
\end{enumerate}
Eliminating the vertices in an order given by a permutation $\pi$
generates a factorization
$ L= P_{\pi} \mathcal{L} \mathcal{D} \mathcal{L}^{\top}
P^{\top}_{\pi},$
where $P_{\pi}$ denotes the permutation matrix of $\pi$, i.e.,
$(P_{\pi}z)_{i} = z_{\pi(i)}$ for all $z.$ Though picking a good order
of elimination can significantly reduce the running time of Cholesky
factorization, it gives no guarantees for general systems, e.g., for
sparse expander graphs, every ordering results in an $\Omega(n^3)$
running time~\cite{LiptonRT79}.


%
%

\paragraph{Our Results.} 
In this paper, we present the first nearly linear time algorithm that
generates a sparse approximate Cholesky decomposition for Laplacian matrices, 
with provable approximation guarantees.
Our algorithm $\cholesky$ can be
described informally as follows (see Section~\ref{sec:algorithm} for a
precise description):
\begin{enumerate}[label={},itemsep=1pt,parsep=0pt,leftmargin=10pt,rightmargin=20pt,topsep=5pt]
\item[] \texttt{Randomly permute the vertices.
  \item[] \FOR $i = 1$ to $n-1$
  \item[] \hspace{\pgmtab} Use equation $i$ to express
the    variable for vertex $i$ in terms of the remaining variables.  
  \item[] \hspace{\pgmtab} Eliminate vertex $i$, adding random 
samples from the clique on the neighbors of $i.$}
\end{enumerate}
We
prove the following theorem about our algorithm, where for symmetric
matrices $A,B,$ we write $A \preceq B$ if $B-A$ is positive
semidefinite (PSD).
\begin{theorem}
\label{thm:intro:main}
The algorithm ${\cholesky},$
given an $n \times n$ Laplacian matrix $L$
with $m$ non-zero entries, runs in expected time $O(m \log^{3} n)$ and
computes a permutation $\pi,$ a lower triangular matrix $\mathcal{L}$
with $O(m \log^{3} n)$ non-zero entries, and a diagonal matrix
$\mathcal{D}$ such that with probability $1-\frac{1}{\poly(n)},$
we have
\[\nfrac{1}{2} \cdot L \preceq 
Z
 \preceq \nfrac{3}{2} \cdot L,\]
where $Z = P_{\pi} \mathcal{L} \mathcal{D} \mathcal{L}^{\top}
P_{\pi}^{\top},$ \emph{i.e.}, $Z$ has a sparse Cholesky factorization. 
\end{theorem}

The sparse approximate Cholesky factorization for $L$
given by Theorem~\ref{thm:intro:main} immediately implies fast
solvers for Laplacian systems. We can use the simplest iterative
method, called iterative refinement~\cite[Chapter 12]{Higham02} to
solve the system $Lx = b $ as follows. We let,
\[x^{(0)} = 0,\qquad  x^{(i+1)} = x^{(i)} - \nfrac{1}{2} \cdot Z^{+} (L
x^{(i)} - b),
\]
where we use the $Z^{+},$ the pseudo-inverse of $Z$ since $Z$ has a
kernel identical to $L.$ 
Let $\vecone$ denote the all ones vector, and for any vector $v,$ let
$\norm{v}_{L} \defeq \sqrt{v^{\top} L v}.$
\begin{theorem}
\label{thm:simple-solver}
For all Laplacian systems $Lx=b$ with $\vecone^{\top} b=0,$ and all
$\epsilon > 0,$ using the sparse approximate Cholesky factorization
$Z$ given by Theorem~\ref{thm:intro:main}, the above iterate for
$t = 3\log \nfrac{1}{\eps}$ satisfies
$\norm{x^{(t)} - L^{+}b}_{L} \le \eps \norm{L^{+}b}_{L}.$ We can
compute such an $x^{(t)}$ in time
$O(m \log^{3} n \log \nfrac{1}{\eps}).$
\end{theorem}


In our opinion, this is the simplest nearly-linear time solver for
Laplacian systems. Our algorithm only uses random sampling, and
no graph-theoretic constructions,
in contrast with all previous
Laplacian solvers. The analysis is also entirely contained in this
paper. We also remark that there is a possibility that our analysis is
not tight, and that the bounds can be improved by a stronger matrix
concentration result. 

\paragraph{Technical Contributions.}
There are several key ideas that are central to our result: The first
is randomizing the order of elimination.  At step $i$ of the
algorithm, if we eliminate a fixed vertex and sample the resulting
clique, we do not know useful bounds on the sample variances that
would allow us to prove concentration.  Randomizing over the choice of
vertex to eliminate allows us to bound the sample variance by roughly
$\nfrac{1}{n}$ times the Laplacian at the previous step.

The second key idea is our approach to estimating effective
resistances: A core element in all nearly linear time Laplacian solvers
is a procedure for estimating effective resistances (or leverage
scores) for edges in order to compute sampling probabilities.
In previous solvers, these estimates are obtained using fairly
involved procedures (e.g. low-stretch trees, ultrasparsifiers, or the
subsampling procedure due to Cohen et al.~\cite{CohenLMMPS15}). In
contrast, our solver starts with the crudest possible estimates of $1$ for
every edge, and then uses the triangle inequality for effective
resistances (Lemma~\ref{lem:levscore-sum-bound}) to obtain estimates
for the new edges generated. We show that these
estimates suffice for 
constructing a nearly linear time
Laplacian solver.

Finally,
we develop new
concentration bounds for a class of matrix martingales that we call
{\seqsum s}. The name is motivated by a simple scalar model: at each
step, we pick a random bag of dice from a collection of bags (in the
algorithm, this corresponds to picking a random vertex to eliminate),
and then we independently roll each die in the selected bag
(corresponding to drawing independent samples from the clique
added). The guarantees obtained from existing matrix concentration
results are too weak for our application. The concentration bound
gives us a powerful tool for handling conditionally independent sums
of variables.  We defer a formal description of the martingales and
the concentration bound to Section~\ref{sec:matrixDefsThms}.

\paragraph{Comparison to other Laplacian solvers.}
Though the current best algorithm for solving a general $n \times n$
positive semidefinite linear system with $m$ non-zero entries takes
time $O(\min\{mn,n^{2.2373}\})$~\cite{Williams12}, a breakthrough result
by Spielman and Teng~\cite{SpielmanT04, SpielmanTengLinsolve}
showed that linear systems in graph Laplacians could be solved in time
$O(m\cdot {\poly (\log n)} \log \nfrac{1}{\eps}).$ There has been a lot of progress over the
past decade~\cite{KMP1,KMP2, KOSZ, CohenKMPPRX, PengS14,
KyngLPSS16}, and the current best running time is
$O(m \log^{\nfrac{1}{2}} n \log \nfrac{1}{\eps})$ (up to ${\polylog}\
n$ factors)~\cite{CohenKMPPRX}. All of these algorithms have relied on
graph-theoretic constructions -- low-stretch trees~\cite{SpielmanT04,
  KMP1, KMP2, KOSZ, CohenKMPPRX}, graph
sparsification~\cite{SpielmanT04, KMP1, KMP2, CohenKMPPRX, PengS14},
and explicit expander graphs~\cite{KyngLPSS16}. 

In contrast, our algorithm requires no graph-theoretic
construction, and is based purely on random sampling. Our result
only uses two algebraic facts about Laplacian
matrices: \begin{enumerate*}
\item They are closed under taking Schur complements, and
\item They satisfy the effective resistance triangle inequality
  (Lemma~\ref{lem:levscore-sum-bound}).
\end{enumerate*}

\cite{KyngLPSS16} presented the first nearly linear time solver for
block Diagonally Dominant (bDD) systems -- a generalization of SDD
systems.  If bDD matrices satisfy the effective resistance triangle
inequality (we are unaware if they do), then the algorithm in the main
body of this paper immediately applies to bDD systems, giving a sparse
approximate block Cholesky decomposition and a
nearly linear time solver for bDD matrices.

In Section~\ref{sec:bdd}, we sketch a near-linear time algorithm for
computing a sparse approximate block Cholesky factorization for bDD
matrices. It combines the approach of {\cholesky} with a recursive
approach for estimating effective resistances, as
in~\cite{KyngLPSS16}, using the subsampling
procedure~\cite{CohenLMMPS15}. Though the algorithm is more involved
than {\cholesky}, it runs in time $O(m \log^{3} n + n \log^{5} n),$ and
produces a sparse approximate Cholesky decomposition with only
$O(m \log^{2} n + n \log^{4} n)$ entries.
 The algorithm only uses that bDD matrices are closed under taking Schur
complements, and that the Schur complements have a clique structure
similar to Laplacians (see Section~\ref{sec:prelims}).





\paragraph{Comparison to Incomplete Cholesky Factorization.}
%
A popular approach to tackling fill-in is 
\emph{Incomplete Cholesky factorization}, where we throw away most of
the new entries generated when eliminating variables. 
The hope is that
the resulting factorization is still an approximation to the original
matrix $L,$ in which case such an approximate factorization can be
used to quickly solve systems in $L.$ 
Though variants of this approach are used often in practice, and we
have approximation guarantees for some families of
Laplacians~\cite{Gustaffson78, Guattery97, BernGHNT06}, there are no known
guarantees for general Laplacians to the best of our
knowledge.

\section{Preliminaries}
\label{sec:prelims}
\paragraph{Laplacians and Multi-Graphs.}
We consider a connected undirected multi-graph $G =(V,E)$, with
positive edges weights $w : E \to \rea_{+}$.  Let $n = \abs{V}$ and
$m = \abs{E}$.  We label vertices $1$ through $n$, s.t.
$V = \setof{1, \ldots, n}$. Let $\ee_{i}$ denote the $i^{\textrm{th}}$
standard basis vector.  Given an ordered pair of vertices $(u,v)$, we
define the pair-vector $b_{u,v} \in \rea^{n}$ as
$\bb_{u,v} = \ee_{v} - \ee_{u}.$ For a multi-edge $e,$ with endpoints
$u,v$ (arbitrarily ordered), we define $\bb_{e} = \bb_{u,v}.$
%
By assigning an arbitrary direction to each multi-edge of G we define
the Laplacian of $G$ as $L = \sum_{e \in E} w(e) \bb_{e} \bb_{e}^{\top}.$
Note that the Laplacian does not depend on the choice of direction for
each edge.
Given a single multi-edge $e$, we refer to $w(e) \bb_{e} \bb_{e}^{\top}$
as the Laplacian of $e$.

A weighted multi-graph $G$ is not uniquely defined by its Laplacian,
since the Laplacian only depends on the sum of the weights of the
multi-edges on each edge.
We want to establish a one-to-one correspondence between a weighted
multi-graph $G$ and its Laplacian $L$, so
from now on, we will consider every Laplacian to be maintained
explicitly as a sum of Laplacians of multi-edges, and we will
maintain this multi-edge decomposition as part of our algorithms.
%
%
%
\begin{fact}
If $G$ is connected, then the kernel of $L$ is the span of 
the vector $\vecone.$
\end{fact}
Let $L^{+}$ denote the pseudo-inverse of $L$.
Let $J \defeq \vecone\vecone^{\top}$.
Then, we define
$
\lapid
\defeq
L L^{+}
= I - \frac{1}{n} J
.$
\paragraph{Cholesky Factorization in Sum and Product Forms.}
We now formally introduce Cholesky factorization. 
Rather than the usual perspective where we factor out lower triangular
matrices at every step of the algorithm, we present an equivalent
perspective where we subtract a rank one term from the current matrix,
obtaining its Schur complement. The lower triangular structure follows
from the fact that the matrix effectively become smaller at every step.

Let $L$ be any symmetric positive-semidefinite matrix.
Let $L(:,i)$ denote the $i^{\text{th}}$ column of $L.$
Using the first equation in the system $L x = b$ to
eliminate the variable $x_{1}$ produces another system
$S^{(1)}x^{\prime} = b^{\prime},$ where $b^{\prime}_1 = 0, x^{\prime}$
is $x$ with $x_{1}$ replaced by 0, and
\[ S^{(1)} \defeq L - \frac{1}{L(1,1)}L(:,1) L(:,1)^{\top},\]
is called the \emph{Schur complement} of L with respect to 1.
The first row and column of
$S^{(1)}$ are identically 0, and thus this is effectively a system in
the remaining $n-1$ variables. Letting
$\alpha_{1} \defeq L(v_{1},v_{1}), \cc_{1} \defeq
\frac{1}{\alpha_{1}} L(:,v_{1}),$ we have $L = S^{(1)} + \alpha_{1}
\cc_{1} \cc_{1}^{\top}.$

%
%

For computing the Cholesky factorization, we perform a sequence of
such operations, where in the $k^{th}$ step, we select an index
$v_{k} \in V\setminus \setof{v_{1}, \ldots, v_{k-1}}$ and eliminate
the variable $v_{k}.$ We define
\begin{align*}
  \alpha_{k} &= S^{(k-1)}(v_{k},v_{k}) \\
  \cc_{k} &= \frac{1}{\alpha_{k}} S^{(k-1)}(:,v_{k}) \\
  S^{(k)} &= S^{(k-1)} - \alpha_{k} \cc_{k}\cc_{k}^{\top}.
\end{align*}
If at some step $k$, 
$S^{(k-1)}(v_{k},v_{k}) = 0$,
then we define $\alpha_{k} = 0$,
and $\cc_{k} = 0$.
Continuing until $k = n-1$,
$S^{(k)}$ will have at most one non-zero entry,
which will be on the $v_{n}$ diagonal.
We define $\alpha_{n} = S^{(k)}$ and $\cc_{n} = \ee_{v_{n}}$.

Let $\calC$ be the 
$n \times n$ matrix
with $\cc_{i}$ as its $i^{\textrm{th}}$ column, 
and $\mathcal{D}$ be the $n \times n$ diagonal matrix 
$\mathcal{D}(i,i) = \alpha_{i}$, then
$ L = \sum_{i=1}^{n} \alpha_{i} \cc_{i} \cc_{i}^{\top} 
= \mathcal{C} \mathcal{D} \mathcal{C}^{\top}.$
Define the permutation matrix $P$ by $P \ee_{i} = \ee_{v_{i}}.$
Letting $\mathcal{L} = P^{\top}\mathcal{C},$ we have
$L = 
 P \mathcal{L} \mathcal{D} \mathcal{L}^{\top} P^{\top}.$
This decomposition is known as Cholesky factorization.
Crucially, 
$\mathcal{L}$ is lower triangular, since
$\mathcal{L}(i,j) = (P^{\top}\cc_{j})(i) = \cc_{j}(v_{i}),$ and for
$i < j,$ we have $\cc_{j}(v_{i}) = 0$.



\paragraph{Clique Structure of the Schur Complement.}
Given a Laplacian $L$,
let $\vstar{L}{v} \in \rea^{n \times n}$ denote the Laplacian corresponding to the 
edges incident on vertex $v$, i.e.
\begin{equation}
\label{eq:starDef}
\vstar{L}{v} \defeq \sum_{e \in E : e \ni v} w(e) \bb_{e} \bb_{e}^{\top}
.
\end{equation}
For example, we denote the first column of $L$ by
$
\begin{pmatrix}
  d \\
  -\aa 
\end{pmatrix}
,$
then
$
L_{1} =
\begin{bmatrix}
d & -\aa^{\top} \\
-\aa & \diag{\aa}
\end{bmatrix}
.
$
We can write the Schur complement 
$S^{(1)}$ as
$S^{(1)} = L-\vstar{L}{v} +\vstar{L}{v}  - \frac{1}{L(v,v)}L(:,v)L(:,v)^{\top}.$
It is immediate that $L-\vstar{L}{v} $ is a Laplacian matrix,
since $L-\vstar{L}{v}  = \sum_{e \in E : e \not\ni v} w(e) \bb_{e}
\bb_{e}^{\top}$.
A more surprising (but well-known) fact is that 
\begin{align}
\label{eq:Cv-def}
C_{v}(L) \defeq \vstar{L}{v}  - \frac{1}{L(v,v)}L(:,v) L(:,v)^{\top}
\end{align}
is also a Laplacian, and its edges form a clique on the neighbors of $v$.
It suffices to show it for $v = 1.$ We write $i \sim j$ to denote $(i,j) \in E.$
Then
\begin{equation}
\label{eq:cliquestructure}
C_{1}(L)
=
L_{1} - \frac{1}{L(1,1)}L(:,1) L(:,1)^{\top}
=
\begin{bmatrix}
0 & \veczero^{\top} \\
\veczero & \diag{\aa} - \frac{\aa\aa^{\top} }{d}
\end{bmatrix}
= 
\sum_{i \sim 1} \sum_{j \sim 1} 
\frac{w(1,i) w(1,j)} 
{d}
\bb_{(i,j)} \bb_{(i,j)}^{\top}
.
\end{equation}
Thus
$S^{(1)}$ is a Laplacian since it is a sum of two Laplacians.
By induction, for all $k,$ $S^{(k)}$ is a Laplacian.



\section{The {\cholesky} Algorithm}
\label{sec:algorithm}
\begin{figure}[h]
\begin{mdframed}
\vspace{\algtopspace}
\captionof{table}{ \label{alg:solver} 
  $\cholesky(\epsilon, L):$ Given an $\epsilon > 0$ and a Laplacian
  $L,$ outputs $(P,\mathcal{L} , \mathcal{D}),$ a sparse approximate Cholesky factorization of $L$ }
    \vspace{\algpostcaptionspace}
    \begin{tight_enumerate}
\item \label{alg:solver:initSplit}
$\samp{S}^{(0)} \leftarrow L$ with edges split into $\rho =
\ceil{12 (1+\delta)^{2} \eps^{-2} \ln^{2} n}$ copies with
$1/\rho$ of the original weight
\item Define the diagonal matrix $\mathcal{D} \assign \matzero_{n
    \times n}$
\item Let $\pi$ be a uniformly random permutation.
  Let $P_{\pi}$ be its permutation matrix, i.e., $(P_{\pi}x)_{i} = x_{\pi_{i}}$
\item \FOR $k = 1$ to $n-1$
        \item \hspace{\pgmtab}
          $\mathcal{D}(\pi(k),\pi(k)) \leftarrow (\pi(k),\pi(k))$ entry of $\samp{S}^{(k-1)}$
         \item \hspace{\pgmtab} 
          $\cc_{k} \leftarrow \pi(k)^{\text{th}} \text{ column of }
          \samp{S}^{(k-1)}$ divided by $\mathcal{D}(\pi(k),\pi(k))$
   if $\mathcal{D}(\pi(k),\pi(k)) \neq 0,$ or zero otherwise
        \item \hspace{\pgmtab}  
          $\samp{C}_{k}\leftarrow
          \csamp(\samp{S}^{(k-1)},\pi(k))$
        \item  \hspace{\pgmtab}  
          $\samp{S}^{(k)} \leftarrow \samp{S}^{(k-1)} - \vstar{\samp{S}^{(k-1)}}{\pi(k)} + \samp{C}_{k}$
\item
\label{alg:cholesky:laststep}
$\mathcal{D}(\pi(n),\pi(n)) \leftarrow \samp{S}^{(n)}$ and
  $\cc_{n} \leftarrow \ee_{\pi(n)}$ 
\item
\label{alg:cholesky:perm}
  $\mathcal{L}  \leftarrow P_{\pi}^{\top}
   \begin{pmatrix}
    \cc_{1} & \cc_{2} & \ldots & \cc_{n}
  \end{pmatrix}
 $
\item
\RETURN
$(P_{\pi},\mathcal{L} , \mathcal{D}) $
      \end{tight_enumerate}
 \end{mdframed}
\end{figure}
\begin{figure}[h]
\begin{mdframed}
\vspace{\algtopspace}
\captionof{table}{ \label{alg:sampler}
  $\sum_{i}  Y_{i} = \csamp(S,v)$ : Returns several i.i.d samples of edges from the
  clique generated after eliminating vertex $v$ from the multi-graph
  represented by $S$ }
    \vspace{\algpostcaptionspace}
\begin{tight_enumerate}
\item \FOR $i \assign  1$ to $\mult{S}{v}$
         \item \label{alg:sampler:wSamp}
           \hspace{\pgmtab}
           Sample $e_{1}$ from list of multi-edges incident on $v$ with
           probability $w(e)/w_{S}(v)$.
         \item \label{alg:sampler:uSamp}
           \hspace{\pgmtab}
           Sample $e_{2}$ uniformly from list of multi-edges incident on  $v$.
         \item 
           \hspace{\pgmtab} 
           \IF $e_{1}$ has endpoints $v,u_{1}$ and $e_{2}$ has
           endpoints $v,u_{2}$ and $u_{1} \neq u_{2}$
         \item \label{alg:sampler:assign}
           \hspace{\pgmtab}\hspace{\pgmtab}
           $Y_{i} \assign
           \frac{w(e_{1}) w(e_{2}) }
           {w(e_{1})+ w(e_{2})}
           \bb_{u_{1},u_{2}} \bb_{u_{1},u_{2}} ^{\top}
           $
           \item \label{alg:sampler:assignZero}
             \hspace{\pgmtab} \ELSE $Y_{i} \assign 0$
\item
\RETURN
$\sum_{i}  Y_{i}$
      \end{tight_enumerate}
 \end{mdframed}
\end{figure}
%
Algorithm~\ref{alg:solver} gives the pseudo-code for our algorithm
$\cholesky$.  Our main result, Theorem~\ref{thm:sparsecholesky} (a
more precise version of Theorem~\ref{thm:intro:main}), shows that the
algorithm computes an approximate sparse Cholesky decomposition
in nearly linear time. We assume the Real RAM model.
We prove the theorem in Section~\ref{sec:analysis}.
\begin{theorem}
  \label{thm:sparsecholesky}
  Given a connected undirected multi-graph
  $G =(V,E)$, with positive edges weights 
  $w : E \to \rea_{+}$, and associated Laplacian $L$,
  and scalars $\delta >1$, $0<\eps\leq1/2$,
  the algorithm $\cholesky(L,\eps,\delta)$
  returns an approximate sparse Cholesky decomposition
  $(P,\mathcal{L},\mathcal{D})$ s.t. 
  with probability at least $1-2/n^{\delta}$,
  \begin{align}
    \label{eq:lapErrorBounds}
   (1-\eps) L
   \preceq
   P \mathcal{L} \mathcal{D} \mathcal{L}^{\top} P^{\top} 
   \preceq
   (1+\eps) L   .
  \end{align}
  The expected number of non-zero entries in $\mathcal{L}$ is $O(\frac{\delta^{2}}{\eps^{2}}
  m\log^{3} n)$.
  The algorithm runs in expected time $O(\frac{\delta^{2} }{\eps^{2}}
  m\log^{3} n)$.
\end{theorem}
Algorithm~\ref{alg:sampler} gives the pseudo-code for our $\csamp$ algorithm.

The most significant obstacle to making Cholesky factorization of
Laplacians efficient is the \emph{fill-in} phenomenon,
namely that each clique $C_{v}(S)$ has roughly $(\mult{S}{v})^{2}$ non-zero entries.
To solve this problem, we develop a sampling procedure $\csamp$ 
that produces a sparse Laplacian matrix which approximates the clique $C_{v}(S)$.
As input, the procedure requires a Laplacian matrix $S$, maintained
as a sum of Laplacians of multi-edges, and a vertex $v$.
It then computes a sampled matrix that approximates $C_{v}(S)$.
The elimination step in $\cholesky$ removes the $\mult{S}{v}$ edges
incident on $v$, and $\csamp(S,v)$ only adds at most $\mult{S}{v}$
multi-edges.
This means the total number of multi-edges does not increase with each
elimination step, solving the \emph{fill-in} problem.
The sampling procedure is also very fast:
It takes $O(\mult{S}{v})$ time, much faster than the
order $(\mult{S}{v})^{2}$ time required to even write down the clique $C_{v}(S)$.

Although it is notationally convenient for us to pass the whole matrix
$S$ to $\csamp$, the procedure only relies on multi-edges incident on
$v$, so we will only pass these multi-edges.


\begin{remark}
\label{rem:timeConc}
Theorem~\ref{thm:sparsecholesky} only provides guarantees only on the
expected running time.
In fact, if we make a small change to the algorithm,
we can get
 $O(\frac{\delta^{2} }{\eps^{2}}
  m\log^{3} n)$ running time w.h.p.
At the $k^{th}$ elimination, instead of picking the
vertex $\pi(k)$ uniformly at random among the remaining vertices,
we pick the vertex uniformly at
random among the remaining vertices with at most twice the average
multi-edge degree in $S^{(k)}$.
In Appendix~\ref{sec:timeConcentration}, we sketch a proof of this.
\end{remark}


\section{Analysis of the Algorithm using Matrix Concentration}
\label{sec:analysis}
In this section, we analyze the $\cholesky$ algorithm,
and prove Theorem~\ref{thm:sparsecholesky}.
To prove the theorem, we need several intermediate results which we
will now present.
In Section~\ref{sec:csampdistrAndMartingale},
we show how the output the $\cholesky$ and $\csamp$ algorithms
can be used to define a matrix martingale.
In Section~\ref{sec:matrixDefsThms}, we introduce a new type of
martingale, called a {\seqsum}, and a novel matrix concentration result for these martingales.
In Section~\ref{sec:csampMomentsAndCholeskyProof}, we show how to
apply our new matrix concentration results to the $\cholesky$
martingale and prove Theorem~\ref{thm:sparsecholesky}.
We defer proofs of the lemmas that characterize $\csamp$ to
Section~\ref{sec:csampProofs},
and proofs of our matrix concentration results to Section~\ref{sec:matrixProofs}.

\subsection{Clique Sampling and Martingales for Cholesky Factorization}
\label{sec:csampdistrAndMartingale}
Throughout this section, we will study matrices that arise in the
when using $\cholesky$ to produce a sparse approximate Cholesky
factorization of 
the Laplacian $L$ of a multi-graph $G$.
We will very frequently need to refer to matrices that are
normalized by $L$.
We adopt the following notation: Given a symmetric matrix $S$
s.t. $\ker(L) \subseteq \ker(S)$,
\[
\lnorm{S} \defeq (L^{+})^{1/2} S  (L^{+})^{1/2}.
\]
We will only use this notation for matrices $S$ that satisfy the
condition $\ker(L) \subseteq \ker(S)$. 
Note that $\lnorm{L} = \lapid,$ and $A \preceq B$ iff $\lnorm{A}
\preceq \lnorm{B}.$
Normalization is always
done with respect to the Laplacian $L$ input to $\cholesky$.
We say a multi-edge $e$ is $1/\rho$-bounded if
\[
\norm{
w(e)  \lnorm{\bb_{e}\bb_{e}^{\top}}
}
\leq
1/\rho
.
\]
Given a Laplacian $S$ that corresponds to a multi-graph $G_{S}$,
and a scalar $\rho > 0$,
we say that $S$ is $1/\rho$-bounded if every multi-edge of 
$S$ is $1/\rho$-bounded.
Since every multi-edge of $L$ is trivially $1$-bounded, 
we can obtain a $1/\rho$-bounded Laplacian that corresponds to the
same matrix, by splitting each multi-edge into $\ceil{\rho}$ identical copies,
with a fraction $1/\ceil{\rho}$ of the initial weight. The
resulting Laplacian has at most
$\ceil{\rho} m$ multi-edges.

Our next lemma describes some basic properties of the samples output by
$\csamp$. We prove the lemma in Section~\ref{sec:csampProofs}.
\begin{lemma}
\label{lem:csampdistr}
Given a Laplacian matrix $S$ that is $1/\rho$-bounded w.r.t. $L$
and a vertex $v$, $\csamp(S,v)$  returns a sum $\sum_{e} Y_{e}$
of $\mult{S}{v}$ IID samples
$Y_{e} \in \rea^{n \times n}$.
The following conditions hold
\begin{tight_enumerate}
\item 
\label{lem:csampdistr:lap}
$Y_{e}$ is $0$
or the Laplacian of a multi-edge with endpoints $u_{1},u_{2}$,
where $u_{1},u_{2}$ are neighbors of $v$ in $S$.
\item $\expt \sum_{e} Y_{e} = C_{v}(S)$.
\item $\norm{\lnorm{Y}_{e}} \leq 1/\rho$, i.e. $Y_{e}$ is 
$1/\rho$-bounded w.r.t. $L$.
\end{tight_enumerate}
The algorithm runs in time $O(\mult{S}{v})$.
\end{lemma}
The lemma tells us that the samples in expectation behave like the 
clique $C_{v}(S)$, and that each sample is $1/\rho$-bounded.
This will be crucial to proving concentration properties of our algorithm.
We will use the fact that the expectation of the $\csamp$
algorithm output equals the matrix produced by standard Cholesky
elimination, to show that in expectation, the sparse approximate
Cholesky decomposition 
produced by our $\cholesky$ algorithm equals the original Laplacian.
We will also see how we can use this expected behaviour to represent
our sampling process as a martingale.
We define the $k^{\textrm{th}}$ approximate Laplacian as 
\begin{equation}
\label{eq:kApxLap}
L^{(k)} = \samp{S}^{(k)} + \sum_{i=1}^{k} \alpha_{i} \cc_{i} \cc_{i}^{\top} 
.
\end{equation}
Thus our final output equals $L^{(n)}$.
Note that Line~\eqref{alg:cholesky:laststep} of the $\cholesky$ algorithm does not introduce any
sampling error, and so $L^{(n)} = L^{(n-1)}$.
The only significance of Line~\eqref{alg:cholesky:laststep} is that it
puts the matrix in the form we need for our factorization.
Now
\begin{align*}
L^{(k)} - L^{(k-1)}
&= \alpha_{k}\cc_{k}\cc_{k}^{\top} +
\samp{S}^{(k)} - \samp{S}^{(k-1)}
\\
&= \alpha_{k}\cc_{k}\cc_{k}^{\top} + \samp{C}_{k} -
  \vstar{\samp{S}^{(k-1)}}{{\pi(k)}} 
\\
&= \samp{C}_{k}
- C_{\pi(k)}(\samp{S}^{(k-1)})
.
\end{align*}
Each call to $\csamp$ returns a sum of sample edges.
Letting $Y^{(k)}_{e}$ denote the $e^{\text{th}}$ sample in the
$k^{\text{th}}$ call to $\csamp$, we can write this sum as 
$\sum_{e} Y^{(k)}_{e}$.
Thus, conditional on the choices of the $\cholesky$ algorithm until
step $k-1$, and conditional on $\pi(k)$,
we can apply Lemma~\ref{lem:csampdistr}
to find that the expected value of $\samp{C}_{k} $ is
$\sum_{e} \expt_{Y^{(k)}_{e}} Y^{(k)}_{e} = C_{\pi(k)}(\samp{S}^{(k-1)})$.
Hence the expected value of $L^{(k)}$ is exactly $L^{(k-1)}$,
 and we can write
\begin{align*}
  L^{(k)} - L^{(k-1)} = \sum_{e} Y^{(k)}_{e} - \expt_{Y^{(k)}_{e}} Y^{(k)}_{e}.
\end{align*}
By defining $X^{(k)}_{e} \defeq Y^{(k)}_{e} - \expt_{Y^{(k)}_{e}}
Y^{(k)}_{e}$, this becomes $L^{(k)} - L^{(k-1)} = \sum_{e}
X^{(k)}_{e}$.
So, without conditioning on the choices of the $\cholesky$
algorithm, we can write
\[
L^{(n)} - L =   L^{(n)} - L^{(0)}
  = \sum_{i = 1}^{n-1} \sum_{e} X^{(i)}_{e}
.
\]
This is a martingale.
To prove multiplicative concentration bounds,
we need to normalize the martingale by $L$, and so instead we
consider
\begin{align}
\label{eq:normalMartingale}
\lnorm{L^{(n)}} - \lnorm{L}
=
\lnorm{L^{(n-1)}} - \lnorm{L}
= \lnorm{L^{(n)}} - \lapid
= \sum_{i = 1}^{n-1} \sum_{e} \lnorm{X^{(i)}_{e}}
.
\end{align}
This martingale has considerable structure beyond a standard
martingale.
Conditional on the choices of the $\cholesky$ algorithm until
step $k-1$, and conditional on $\pi(k)$, the terms
$\lnorm{X^{(k)}_{e}}$
are independent.

In Section~\ref{sec:matrixDefsThms} we define a type of martingale
that formalizes the important %
aspects of this structure.


\subsection{{\SEQSUM}s and Matrix Concentration Results}
\label{sec:matrixDefsThms}
We use the following shorthand notation:
Given a sequence of random variables 
$(r_{1}, R^{(1)},r_{2},R^{(2)},\ldots, r_{k},R^{(k)}) $,
 for every $i,$ we
 write 
\[\seqexpt{i} \sqbr{\ \cdot\ } = \expt_{r_{1}} 
\expt_{R^{(1)}}
\cdots
\expt_{r_{i}} 
\expt_{R^{(i)}} 
\sqbr{\ \cdot\ }
.\]
Extending this notation to conditional
expectations, we write,
\[ \expt_{r_{i} } \sqbr{\ \cdot\ \middle| (i-1) }  = \expt_{r_{i}} \sqbr{\
  \cdot\ \middle \vert r_{1},R^{(1)},\ldots,r_{i-1},R^{(i-1)}}\]
\begin{definition}
  \textbf{ A {\seqsum}} is a sum of random 
  $d \times d$ matrices ${Z = \sum_{i = 1}^{k} \sum_{e = 1}^{l_{i}}Z^{(i)}_{e}}$,
  with some additional structure.
  We require that there exists a sequence of random variables
  $(r_{1}, R^{(1)},r_{2},R^{(2)},\ldots, r_{k},R^{(k)}) $, s.t. for all $1
  \leq i \leq k$, conditional on $(i-1)$ and $r_{i}$,
   $l_{i}$ is a non-negative integer, 
   and $R^{(i)}$ is a tuple of $l_{i}$ independent random variables:
   $R^{(i)} = (R^{(i)}_{1}, \ldots, R^{(i)}_{l_{i}})$.
   We also require that conditional on $(r_{1},
   R^{(1)},r_{2},R^{(2)},\ldots, r_{i},R^{(i)})$
   and $r_{i}$,
   for all  $1 \leq e \leq l_{i}$ 
   $Z^{(i)}_{e}$ is a symmetric matrix, and a
   deterministic function of $R^{(i)}_{e}$.
   Finally, we require that
   $\expt_{R^{(i)}_{e} } \sqbr{ Z^{(i)}_{e}
      \middle| (i-1), r_{i} }  = 0$.
\end{definition}
Note that $l_{i}$ is allowed to be random, as long as it is fixed 
conditional on $(i-1)$ and $r_{i}$.
The martingale given in Equation~\eqref{eq:normalMartingale}
is a {\seqsum}, with $r_{i} = \pi(i)$, and $R^{(i)}_{e} =
Y^{(i)}_{e}$.
The name is motivated by a simple model: At each step
of the martingale we pick a random bag of dice from a collection of bags
(this corresponds to the outcome of $r_{i}$) and then we independently roll each die in
the bag (corresponds to the outcomes $R^{(i)}_{e}$).

It is instructive to compare the {\seqsum}s with
matrix martingales such as those considered by Tropp~\cite{tropp2012user}.
A naive application of the results from~\cite{tropp2012user} tells us that if we have good uniform
norm bounds on each term $Z^{(i)}_{e}$,
and there exists fixed matrices $\Omega^{(i)}_{e}$ s.t. that for
all $i,e$ and for all possible outcomes we get deterministic bounds on
the matrix variances:
$\expt_{R^{(i)}_{e}} (Z^{(i)}_{e})^2 \preceq \Omega^{(i)}_{e}$,
then the concentration of the martingale is governed by
the norm of the sum of the variances
 $\norm{ \sum_{i} \sum_{e} \Omega^{(i)}_{e}} $.
In our case, this result is much too weak: Good enough $\Omega^{(i)}_{e}$
do not exist.

A slight extension of the results from~\cite{tropp2012user} allows us
to do a somewhat better: Since the terms $Z^{(i)}_{e}$ are independent
conditional on $r_{i}$, it suffices to produce fixed matrices
$\Omega^{(i)}$ s.t. that for all $i,e$ and for all possible outcomes
we get deterministic bounds on the matrix variances of the sum of
variables in each ``bag'':
$\sum_{e} \expt_{R^{(i)}_{e}} (Z^{(i)}_{e})^2 \preceq \Omega^{(i)}$.
Then the concentration of the martingale is governed by
$\norm{ \sum_{i} \Omega^{(i)}_{}} $.  Again, this result is not strong
enough: Good fixed $\Omega^{(i)}$ do not seem to exist.

We show a stronger result: If we can produce a uniform norm bound on 
$\expt_{R^{(i)}_{e}} (Z^{(i)}_{e})^2$,
then it suffices to produce fixed matrices $\widehat{\Omega}^{(i)}$ that upper
bound the matrix variance \emph{averaged over all bags}:
${\expt_{r_{i}} \sum_{e} \expt_{R^{(i)}_{e}} (Z^{(i)}_{e})^2 \preceq \widehat{\Omega}^{(i)}}$.
The concentration of the martingale is then governed by $\norm{ \sum_{i} \widehat{\Omega}^{(i)}_{}} $.
In the context where we apply our concentration result,
our estimates of $\norm{\sum_{i} \Omega^{(i)}_{}} $ are larger than
$\norm{\sum_{i} \widehat{\Omega}^{(i)}_{}} $
by a factor $\approx \frac{n}{\ln n}$.
Consequently, we would not obtain strong enough concentration using the weaker result.

The precise statement we prove is:
\begin{theorem}
\label{thm:smoothableTailProb}
Suppose ${Z = \sum_{i = 1}^{k} \sum_{e = 1}^{l_{i}}Z^{(i)}_{e}}$
is a  {\seqsum} of $d \times d$ matrices
that satisfies
\begin{tight_enumerate}
\item Every sample $Z^{(i)}_{e}$ satisfies
  $\norm{Z^{(i)}_{e}}^{2} \le \sigma^{2}_{1},$
\item For every $i$ we have
  $\norm{\sum_{e} \expt_{R^{(i)}_{e}} \sqbr{(Z^{(i)}_{e})^2 \middle|
      (i-1),r_{i}}} \le \sigma^{2}_{2},$ and
\item There exist deterministic matrices $\Omega_{i}$ such that
  $\expt_{r_{i}} \sqbr{\sum_{e} \expt_{R^{(i)}_{e}} (Z^{(i)}_{e})^2\middle| (i-1)}
\preceq \Omega_{i},$
and $\norm{\sum_{i} \Omega_{i}} \le \sigma^{2}_{3}.$
\end{tight_enumerate}
Then, for all $\eps > 0,$ we have
  \[
  \Pr\left[
   Z \not\preceq \eps I
  \right]
  \leq
  \exp \paren{ -\eps^2/4\sigma^2}, 
  \]
where 
\[\sigma^2 = \max\curlbr{\sigma^{2}_{3}, \frac{\epsilon}{2}
  \sigma_{1},\frac{4\epsilon}{5}
  \sigma_{2} }.\]
\end{theorem}
We remark that this theorem, and all the results in this section
extend immediately to Hermitian matrices. We prove the above theorem in Section~\ref{sec:matConc}.
This result is based on the techniques introduced by
Tropp~\cite{tropp2012user} for using Lieb's Concavity Theorem to prove
matrix concentration results.
Tropp's result improved on earlier work, such as Ahlswede and Winter
\cite{ahlswede2002strong} and Rudelson and
Vershynin~\cite{Rudelson2007}. These earlier matrix concentration
results required IID sampling, making them unsuitable for our purposes.

Unfortunately, we cannot apply Theorem~\ref{thm:smoothableTailProb}
directly to the {\seqsum} in Equation~\eqref{eq:normalMartingale}.
As we will see later, the variance of $\sum_{e}\lnorm{X^{(i)}_{e}}$
can have norm proportional to $\norm{\lnorm{L^{(i)}}}$, which can grow
large.

However, we expect that the probability of $\norm{\lnorm{L^{(i)}}}$
growing large is very small.
Our next construction allows us to leverage this idea,
and avoid the small probability tail events that prevent us from
directly applying Theorem~\ref{thm:smoothableTailProb}
to the {\seqsum} in Equation~\eqref{eq:normalMartingale}.
\begin{definition}
  Given a {\seqsum} of
  $d \times d$ matrices ${Z = \sum_{i = 1}^{k} \sum_{e = 1}^{l_{i}}Z^{(i)}_{e}}$,
  and a scalar $\eps > 0$,
  we define for each $h \in \setof{1,2, \ldots, k+1}$
  the event
 \[
 \goodUntil_{h} = \left[\forall 1 \leq j < h.
   \sum_{i = 1}^{j} \sum_{e = 1}^{l_{i}}Z^{(i)}_{e}
  \preceq \eps I \right].
 \]
 We also define the \textbf{$\eps$-\truncSum}:
\[
\truncZ =
\sum_{i = 1}^{k} \left( \ind{\goodUntil_{i}} \sum_{e = 1}^{l_{i}}Z^{(i)}_{e}\right)
\]
\end{definition}
The {\truncSum} is derived from another martingale by
forcing the martingale to get ``stuck'' if it grows too large.
This ensures that so long as the martingale is not stuck, it is not
too large.
On the other hand, as our next result shows, 
the truncated martingale fails more often than the
original martingale, and so it suffices to prove concentration of the
truncated martingale to prove concentration of the original
martingale. The theorem stated below is proven in Section~\ref{sec:matConc}.
\begin{theorem}
\label{thm:truncSum}
  Given a {\seqsum} of
  $d \times d$ matrices ${Z = \sum_{i = 1}^{k} \sum_{e = 1}^{l_{i}}Z^{(i)}_{e}}$,
  a scalar $\eps > 0$,
the associated {$\eps$-\truncSum} $\truncZ$
  is also a {\seqsum}, and
\[
\Pr[ -\eps I \not\preceq Z
\OR
Z \not\preceq \eps I 
]
\leq
\Pr[ -\eps I \not\preceq \truncZ 
\OR
\truncZ \not\preceq \eps I 
]
\]
\end{theorem}

%

%
\subsection{Analyzing the $\cholesky$ Algorithm Using {\SEQSUM}s}
\label{sec:csampMomentsAndCholeskyProof}
By taking $Z^{(k)}_{e} = \lnorm{X^{(k)}_{e}}$,
$r_{i} = \pi(i)$ and $R^{(i)}_{e} = Y^{(i)}_{e}$,
we obtain a {\seqsum}
$Z = \sum_{i = 1}^{n-1} \sum_{e} Z^{(i)}_{e}$. 
Let $\trunc{Z}$ denote the corresponding $\eps$-truncated
{\seqsum}.
The next lemma shows that $\trunc{Z}$ is well-behaved.
The lemma is proven in Section~\ref{sec:csampProofs}.
\begin{lemma}
\label{lem:csampmoments}
Given an integer $1 \leq k \leq n-1$,
conditional on the choices of the $\cholesky$ algorithm until step $k-1$,
let $\sum_{e} Y_{e} = \csamp(\samp{S}^{(k-1)},{\pi(k)})$.
Let ${X_{e} \defeq Y_{e} - \expt_{Y_{e} } Y_{e} }$.
The following statements hold
\begin{tight_enumerate}
\item Conditional on $\pi(k)$, $\expt_{Y_{e}} \ind{\goodUntil_{k}} \lnorm{X_{e}} = 0$.
\item $\norm{\ind{\goodUntil_{k}} \lnorm{X}_{e}} \leq 1/\rho$ holds always.
\item Conditional on $\pi(k)$, $\sum_{e}\expt_{Y_{e}} (\ind{\goodUntil_{k}} \lnorm{X}_{e})^{2}  \preceq
    \ind{\goodUntil_{k}}
        \frac{1}
    {\rho}
    \lnorm{C_{v}(S)}
    .
    $
\item
\label{lem:csampmoments:cexpt}
$\norm{ \ind{\goodUntil_{k}} \lnorm{C_{\pi(k)}(\samp{S}^{(k-1)})} }
\leq 1+\eps$ holds always.
\item $ \expt_{\pi(k)} \ind{\goodUntil_{k}} \lnorm{C_{\pi(k)}(\samp{S}^{(k-1)})} 
  \preceq
   \frac{2(1+\eps)}
    {n+1-k} 
    I
.
$
\end{tight_enumerate}
\end{lemma}
We are now ready to prove Theorem~\ref{thm:sparsecholesky}.
\begin{proofof}{of Theorem~\ref{thm:sparsecholesky}}
We have
$\lnorm{L^{(n)}} =  \lapid  + Z$.
Since for all $k,e$, 
 $\ker(L) \subseteq \ker(Y^{(k)}_{e})$,
the statement
$(1-\eps)L \preceq L^{(n)} \preceq (1+\eps) L$
is equivalent to $-\eps\lapid \preceq Z \preceq \eps\lapid$.
Further, $\lapid Z \lapid = Z$, and so
it is equivalent to $-\eps I \preceq Z \preceq \eps I$.
By Theorem~\ref{thm:truncSum}, we have,
\begin{align}
\label{eq:lapErrorProb}
  \Pr[ (1-\eps)L \preceq L^{(n)} \preceq (1+\eps) L]
  = 1 - \Pr\left[ -Z \not\preceq \eps I 
\OR  Z \not\preceq \eps I \right]
\nonumber
& \geq
  1 - \Pr\left[ -\trunc{Z} \not\preceq \eps I
\OR  \trunc{Z}  \not\preceq\eps I \right]
\nonumber
\\  & \geq
   1 - \Pr\left[ -\trunc{Z} \not\preceq\eps I \right]
    - \Pr\left[ \trunc{Z}  \not\preceq \eps I \right]
\end{align}
To lower bound this probability, we'll prove concentration using
Theorem~\ref{thm:smoothableTailProb}. We now compute the parameters
for applying the theorem.
From Lemma~\ref{lem:csampdistr}, for all $k$ and $e,$ we have
$\expt_{Y_{e}} \ind{\goodUntil_{k}}\lnorm{X^{(k)}_{e}} = 0$,
$\norm{\ind{\goodUntil_{k}} \lnorm{X^{(k)}_{e}}} \leq
\frac{1}{\rho}$. 
Thus, we can pick $\sigma_{1} = \frac{1}{\rho}.$
Next, again by Lemma~\ref{lem:csampdistr}, we have
\[\small
\norm{ \sqbr{\sum_{e} \expt_{Y^{(k)}_{e}}
\paren{\ind{\goodUntil_{k}}\lnorm{X^{(k)}_{e}}}^2 \middle|
    (k),r_{k}}}
\le \frac{1}{\rho} \norm{\ind{\goodUntil_{k}} \lnorm{C_{\pi(k)}(S)}}
\le \frac{1+\eps}{\rho} \le \frac{3}{2\rho}
.
\]
Thus, we can pick $\sigma_{2} = \sqrt{\frac{3}{2\rho}}.$
Finally, Lemma~\ref{lem:csampdistr} also gives,
\begin{align*}
\expt_{\pi(k)} \sum_{e} \expt_{Y^{(k)}_{e}} \paren{\ind{\goodUntil_{k}}\lnorm{ X^{(k)}_{e}}}^2 \preceq
\frac{1}{\rho}
\expt_{\pi(k)} \ind{\goodUntil_{k}}\lnorm{C_{\pi(k)}(\samp{S}^{(k-1)}) }
\preceq
\frac{2(1+\eps)}{\rho (n+1-k)} I
\preceq 
\frac{3}{\rho (n+1-k)} I.
\end{align*}
Thus, we can pick $\Omega_{k} =  \frac{3}{\rho (n+1-k)} I,$ and
\[\sigma_{3}^{2} = \frac{3\ln n}{\rho} \ge \sum_{k=1}^{n-1}
\frac{3}{\rho (n+1-k)}.\]
%
Similarly, we obtain concentration for
$-\trunc{Z}$ with the same parameters.
%
%
Thus, by Theorem~\ref{thm:smoothableTailProb},
\[
\Pr\left[ -\truncZ \not\preceq \eps I \right]
    + \Pr\left[ \truncZ \not\preceq \eps I \right]
\leq
2 n \exp\paren{\nfrac{- \eps^2}{4\sigma^2}},\]
for 
\[\sigma^2 = \max\curlbr{\sigma^{2}_{3}, \frac{\epsilon}{2}
  \sigma_{1},\frac{4\epsilon}{5} \sigma_{2} } = \max\curlbr{\frac{3\ln
    n}{\rho},\frac{\eps}{2}\sqrt{\frac{3}{2\rho}},\frac{4\eps}{5\rho}}.\]

Picking $\rho = \ceil{12 (1+\delta)^{2} \eps^{-2} \ln^{2} n},$
we get $\sigma^{2} \le \frac{\eps^2}{4 (1+\delta)\ln n},$ and 
\[
\Pr\left[ -\truncZ \not\preceq \eps I \right]
    + \Pr\left[ \truncZ  \not\preceq \eps I \right]
\leq 
2n \exp \paren{ - \nfrac{\eps^2}{4\sigma^2}} = 2n
\exp \paren{-(1+\delta) \ln n} = 2n^{-\delta}.
\]
Combining this with Equation~\eqref{eq:lapErrorProb}
establishes Equation~\eqref{eq:lapErrorBounds}.

Finally, we need to bound the expected running time of the algorithm.
We start by observing that the algorithm maintains the two following
invariants:
\begin{tight_enumerate}
\item \label{enum:rhoEdges}
Every multi-edge in $\samp{S}^{(k-1)}$ is $1/\rho$-bounded.
\item
\label{enum:multiEdgeCount}
The total number of multi-edges is at most $\rho m$.
\end{tight_enumerate}
We establish the first invariant inductively.
The invariant holds for $\samp{S}^{(0)}$, because of the splitting of
original edges into $\rho$ copies with weight $1/\rho$.
The invariant thus also holds for $\samp{S}^{(0)} -
\vstar{\samp{S}^{(0)}}{\pi(1)}$, since the multi-edges of this
Laplacian are a subset of the previous ones. By
Lemma~\ref{lem:csampdistr}, every multi-edge $Y_{e}$ output by
$\csamp$ is $1/\rho$-bounded, so
$\samp{S}^{(1)} = \samp{S}^{(0)} - \vstar{\samp{S}^{(0)}}{\pi(1)} +
\samp{C}_{1}$ is $1/\rho$-bounded.
If we apply this argument repeatedly for $k=1,\ldots,n-1$ we get invariant~\eqref{enum:rhoEdges}.

Invariant~\eqref{enum:multiEdgeCount} is also very simple to establish:
It holds for $\samp{S}^{(0)}$, because splitting of
original edges into $\rho$ copies does not produce more than $\rho m$
multi-edges in total.
When computing $\samp{S}^{(k)}$, we subtract
$\vstar{\samp{S}^{(k-1)}}{\pi(k)}$, which removes exactly
$\mult{\samp{S}^{(k-1)}}{\pi(k)}$ multi-edges, while we add the
multi-edges produced by the call to $\csamp(\samp{S}^{(k-1)},\pi(k))$,
which is at most $\mult{\samp{S}^{(k-1)}}{\pi(k)}$.
So the number of multi-edges is not increasing.

By Lemma~\ref{lem:csampdistr},
the running time for the call to $\csamp$
is $O(\mult{\samp{S}^{(k)}}{\pi(k)})$.
Given the invariants,
we get that the expected time for the $k^{th}$ call to $\csamp$ is
$O(\expt_{\pi(k)} \mult{\samp{S}^{(k)}}{\pi(k)} ) = O(\rho m / (n+1-k))$.
Thus the expected running time of all calls to $\csamp$ is
$O( m \rho \sum_{i = 1}^{n-1} \frac{1}{n-i}) = O( m \delta^{2}
\eps^{-2} \ln^{3} n )$.
The total number of entries in the $\mathcal{L},\mathcal{D}$ matrices
must also be bounded by $O( m \delta^{2} \eps^{-2} \ln^{3} n )$ in expectation,
and so the permutation step in Line~\ref{alg:cholesky:perm} can be applied in 
expected time
$O( m \delta^{2}\eps^{-2} \ln^{3} n )$,
and this also bounds the expected running time of the whole algorithm.
\end{proofof}


\section{Clique Sampling Proofs}
\label{sec:csampProofs}
In this section, we prove Lemmas~\ref{lem:csampdistr} and
\ref{lem:csampmoments} that characterize the
behaviour of our algorithm $\csamp$,
which is used in $\cholesky$ to approximate the clique generated by
eliminating a variable.

A important element of the $\csamp$ algorithm is our very simple approach to
leverage score estimation.
Using the
well-known result that effective resistance in Laplacians is a
distance (see Lemma~\ref{lem:reffDist}),
we give a bound on the leverage scores of all edges in a clique that
arises from elimination.
We let
\[
w_{S}(v) = \sum_{\substack{e \in E(S) \\ e \ni v}} w(e)
.
\]
Then by Equation~\eqref{eq:cliquestructure}
\begin{equation}
C_{v}(S)
= 
\frac{1}{2}
\sum_{\substack{
e \in E(S) \\
e \text{ has}\\
\text{endpoints}\\
v,u
}}
\sum_{\substack{
e' \in E(S) \\
e' \text{ has}\\
\text{endpoints}\\
v,z \neq u
 }}
\frac{w(e) w(e')} 
{w_{S}(v)}
\bb_{u,z} \bb_{u,z}^{\top}
.
\end{equation}
Note that the factor $\nfrac{1}{2}$ accounts for the fact that every pair is
double counted.
%
\begin{lemma}
\label{lem:levscore-sum-bound}
Suppose multi-edges $e,e' \ni v$
are $1/\rho$-bounded w.r.t. $L$,
and have endpoints $v,u$ and $v,z$ respectively,
and $z \neq u$,
then $w(e)w(e') \bb_{u,z} \bb_{u,z}^{\top}$ is $\frac{w(e) + w(e') }{\rho}$-bounded.
\end{lemma}
To prove Lemma~\ref{lem:levscore-sum-bound},
we need the following result about Laplacians:
\begin{lemma}
\label{lem:reffDist}
  Given a connected weighted multi-graph $G = (V,E,w)$ with
  associated Laplacian matrix $L$ in $G$,
  consider three distinct vertices $u,v,z \in V$,
  and the pair-vectors $b_{u,v}$, $b_{v,z}$ and $b_{u,z}$.
\[
\norm{\lnorm{\bb_{u,z} \bb_{u,z}^{\top}}}
\leq 
\norm{\lnorm{\bb_{u,v}\bb_{u,v}^{\top}}}
+
\norm{\lnorm{\bb_{v,z}\bb_{v,z}^{\top}}}
.
\]
\end{lemma}
This is known as phenomenon that Effective Resistance is a distance~\cite{klein1993resistance}.
\begin{proofof}{of Lemma~\ref{lem:levscore-sum-bound}}
Using the previous lemma:
\begin{align*}
w(e) w(e')
\norm{
\lnorm{\bb_{u,z} \bb_{u,z}^{\top}}
}
 \leq
w(e) w(e')
\left(
\norm{\lnorm{\bb_{u,v} \bb_{u,v}^{\top}}}
\right.
 +
\left.
\norm{\lnorm{\bb_{v,z} \bb_{v,z}^{\top}}}
\right)
\leq
\frac{1}
{\rho}
\left(
w(e) + w(e')
\right)
\end{align*}
\end{proofof}
To prove Lemma~\ref{lem:csampdistr}, we need the following result
of Walker~\cite{Walker77} (see Bringmann and
Panagiotou~\cite{bringmann2012efficient} for a modern statement of the result).
\begin{lemma}
\label{lem:fastdiscretesamp}
Given a vector $p \in \rea^{d}$ of non-negative values,
 the
procedure $\textsc{UnsortedProportionalSampling}$ requires $O(d)$ preprocessing time and after
this allows for IID sampling for a random variable $x$ distributed
s.t.
\[
\Pr[ x = i ] = p(i)/\norm{p}_{1}.
\]
The time required for each sample is $O(1)$.
\end{lemma}
\begin{remark}
  We note that there are simpler sampling constructions than that
  of Lemma~\ref{lem:fastdiscretesamp} that need $O(\log n)$ time per
  sample, and using such a method would only worsen our running time
  by a factor $O(\log n)$.
\end{remark}
\begin{proofof}{of Lemma~\ref{lem:csampdistr}}
From Lines~\eqref{alg:sampler:assign} and~\eqref{alg:sampler:assignZero},
$Y_{i}$ is $0$
or the Laplacian of a multi-edge with endpoints $u_{1},u_{2}$.
To upper bound the running time, 
it is important to note that we do \emph{not} need access to the
entire matrix $S$.
We only need the multi-edges incident on $v$.
When calling $\csamp$, we only pass a copy of just these
multi-edges.

We observe that the uniform samples in Line~\eqref{alg:sampler:uSamp}
can be done in $O(1)$ time each,
provided we count the number of multi-edges incident on $v$ to find
$\mult{S}{v}$.
We can compute $\mult{S}{v}$ in $O(\mult{S}{v})$ time.
Using Lemma~\ref{lem:fastdiscretesamp},
if we do $O(\mult{S}{v})$ time preprocessing,
we can compute each sample in Line~\eqref{alg:sampler:wSamp} in time $O(1)$.
Since we do $O(\mult{S}{v})$ samples, the total time for sampling is
hence $O(\mult{S}{v})$.

Now we determine the expected value of the sum of the
samples.
Note that in the sum below,
each pair of multi-edges appears twice,
with different weights.
\begin{align*}
\expt \sum_{i} Y_{i}
& = 
\mult{S}{v}
\sum_{\substack{
e \in E(S) \\
e \text{ has}\\
\text{endpoints}\\
v,u
}}
\sum_{\substack{
e' \in E(S) \\
e' \text{ has}\\
\text{endpoints}\\
v,z \neq u
 }}
  \frac{w(e)}
  {w_{S}(v)}
  \frac{1}
  {\mult{S}{v}}
  \frac{w(e) w(e') }
  {w(e) + w(e')}
\bb_{u,z} \bb_{u,z}^{\top}
=
C_{v}(S)
.
\end{align*}
By Lemma~\ref{lem:levscore-sum-bound},
\[
\norm{\lnorm{Y_{i}}} 
\leq
\max_{
\substack{
e, e' \in E(S) \\
e,e' \text{ has}\\
\text{endpoints}\\
v,u \text{ and } v,z \neq u
 }
}
\frac{w(e) w(e')}
 {w(e)+ w(e')}
 \norm{     
   \lnorm{\bb_{u,z} \bb_{u,z}^{\top}}
 }
\leq
1/\rho
.
\]
\end{proofof}

\begin{proofof}{of Lemma~\ref{lem:csampmoments}}
Throughout the proof of this lemma,
all the random variables considered are
conditional on the choices of the $\cholesky$ algorithm up to and including step
$k-1$.

Observe, by Lemma~\ref{lem:csampdistr}:
\[
\sum_{e} \expt_{Y_{e}} \lnorm{Y_{e}}^{2}
\preceq
\expt \sum_{e}
\expt_{Y_{e}} 
\norm{\lnorm{Y_{e}}} \lnorm{Y_{e}}
\preceq
\frac{1}{\rho}
\lnorm{C_{\pi(k)}(\samp{S}^{(k-1)})}
.
\]
Now
  \[
  \expt_{Y_{e}}  \ind{\goodUntil_{k}} \lnorm{X_{e}} = \ind{\goodUntil_{k}} \expt_{Y_{e}}  \lnorm{X_{e}}
  = 0
  .
  \]
  Note that $\lnorm{Y_{e}}$ and $\expt_{Y_{e}} \lnorm{Y_{e}}$ are PSD,
  and $\norm{\expt_{Y_{e}} \lnorm{Y_{e}}} \leq \expt_{Y_{e}} \norm{\lnorm{Y_{e}}} \leq 1/\rho$
  so
  \[
  \norm{\lnorm{X_{e}}} = \norm{\lnorm{Y_{e}} - \expt_{Y_{e}} \lnorm{Y_{e}}} 
  \leq \max\setof{\norm{\lnorm{Y_{e}}},\norm{\expt_{Y_{e}} \lnorm{Y_{e}}}} \leq
  1/\rho
  .
  \]
  Also $\expt_{Y_{e}}  (\ind{\goodUntil_{k}} \lnorm{X_{e}})^{2}
   = \ind{\goodUntil_{k}} \expt_{Y_{e}} \lnorm{X_{e}}^{2}$,
  so
  \begin{align*}
    \expt_{Y_{e}}  \lnorm{X_{e}}^{2}
     =
      (\expt_{Y_{e}}  \lnorm{Y_{e}}^{2}) - (\expt_{Y_{e}}  \lnorm{Y_{e}})^{2}
       \preceq 
       (\expt_{Y_{e}}  \lnorm{Y_{e}}^{2}),
  \end{align*}
where, in the last line, we used $0 \preceq (\expt_{Y_{e}}  \lnorm{Y_{e}})^{2}$.
Thus
$\sum_{e} \expt_{Y_{e}}  (\ind{\goodUntil_{i}} \lnorm{X_{e}})^{2}
   \preceq
   \ind{\goodUntil_{k}}
    \frac{1}
    {\rho}
    \lnorm{C_{\pi(k)}(\samp{S}^{(k-1)})}$.
Equation~\eqref{eq:normalMartingale} gives:
\begin{align}
  L^{(k-1)}
  = L + \sum_{i = 1}^{k-1} \sum_{e} X^{(i)}_{e}
.
\end{align}
Consequently, $\ind{\goodUntil_{k}} = 1$ gives 
$ L^{(k-1)} \preceq (1+\eps) L$.

Now, $C_{\pi(k)}(\samp{S}^{(k-1)})$
is PSD since it is a Laplacian, so
$\norm{ \lnorm{C_{\pi(k)}(\samp{S}^{(k-1)})}} =
\lambda_{\max}( \lnorm{C_{\pi(k)}(\samp{S}^{(k-1)})} )$.
By Equation~\eqref{eq:Cv-def}, we get
${C_{\pi(k)}(\samp{S}^{(k-1)})
\preceq
\vstar{\samp{S}^{(k-1)}}{\pi(k)}}
$
and by Equation~\eqref{eq:starDef}
we get
${\vstar{\samp{S}^{(k-1)}}{\pi(k)} \preceq \samp{S}^{(k-1)}}$,
finally by Equation~\eqref{eq:kApxLap}
we get
${\samp{S}^{(k-1)}
\preceq L^{(k-1)}}$
so
\[
\norm{\ind{\goodUntil_{k}} \lnorm{ C_{\pi(k)}(\samp{S}^{(k-1)}) }}
\leq
\ind{\goodUntil_{k}}\lambda_{\max}(\lnorm{ C_{\pi(k)}(\samp{S}^{(k-1)}) })
\leq
\ind{\goodUntil_{k}}\lambda_{\max}(\lnorm{ L^{(k-1)} })
\leq 1+\eps
.
\]
Again, using
$C_{\pi(k)}(\samp{S}^{(k-1)})
\preceq
\vstar{\samp{S}^{(k-1)}}{\pi(k)}$,
we get
\begin{align}
\label{eq:randCliqueExpectation}
\expt_{\pi(k)} \ind{\goodUntil_{k}} 
\lnorm{ C_{\pi(k)}(\samp{S}^{(k-1)}) }
\preceq
\ind{\goodUntil_{k}} \expt_{\pi(k)} 
\lnorm{ \vstar{\samp{S}^{(k-1)}}{\pi(k)} }
=
\ind{\goodUntil_{k}} 
\frac{2}
{n+1-k}
 \lnorm{S^{(k-1)}} 
\preceq
\frac{2(1+\eps)}
    {n+1-k} 
    I
.
\end{align}
\end{proofof}



\section{Matrix Concentration Analysis}
\label{sec:matConc}
\label{sec:matrixProofs}
\subsection{Concentration of {\SEQSUM}s}

To prove Theorem~\ref{thm:smoothableTailProb}, we need the following lemma, which is the main
technical result of this section.
\begin{lemma}
\label{lem:smoothableMGF}
Given a 
  {\seqsum} of $d \times d$ matrices
  ${Z = \sum_{i = 1}^{k} \sum_{e = 1}^{l_{i}}Z^{(i)}_{e}}$ that is 
for all $\theta$ such that $0 < \theta^{2} \le
\min \{ \frac{1}{\sigma_{1}^{2}}, \frac{5}{12\sigma_{2}^{2}}\}$, we have,
 \[
   \expt{  \trop{\exp \paren{\theta Z}} }
  \leq
d\exp\paren{\theta^2 \sigma_{3}^{2} }.
  \]
\end{lemma}
Before proving this lemma, we will see how to use it to prove
Theorem~\ref{thm:smoothableTailProb}.
\begin{proofof}{of Theorem~\ref{thm:smoothableTailProb}}
Given Lemma~\ref{lem:smoothableMGF},
we can show Theorem~\ref{thm:smoothableTailProb}
using the following bound via trace exponentials, which was first developed by
Ahlswede and Winter \cite{ahlswede2002strong}.
\begin{align*}
\prob{}{ Z \not\preceq \eps I }
=
\prob{}{ \lambda_{\max}(Z) \geq \eps }
=
\prob{}{ \lambda_{\max}(\exp \paren{\theta Z}) \geq \exp \paren{\theta \eps }}
& \leq
\exp \paren{- \eps \theta }  \expt \lambda_{\max}(\exp \paren{\theta Z}) 
\\
& \leq 
\exp \paren{- \eps \theta } \expt \trop {\exp \paren{\theta Z}} 
\\
& 
\leq 
 d\exp\paren{-\eps \theta+\theta^2 \sigma_{3}^{2} }.
\end{align*}
Picking $\theta = \frac{\eps}{2\sigma^{2}},$ where
$\sigma^2 = \max\curlbr{\sigma^{2}_{3}, \frac{\epsilon}{2}
  \sigma_{1},\frac{4\epsilon}{5}
  \sigma_{2} },$ the condition for Lemma~\ref{lem:smoothableMGF}
that $0 < \theta^{2} \le
\min \{ \frac{1}{\sigma_{1}^{2}}, \frac{5}{12\sigma_{2}^{2}}\}$ is
satisfied, and we get,
$\prob{}{ Z \not\preceq \eps I }
= d\exp \paren{-\frac{\eps^{2}}{2\sigma^{2}} \paren{1-
  \frac{\sigma_{3}^{2}}{2\sigma^{2}}}} \le d \exp \paren{-\frac{\eps^2}{4\sigma^2}}.$
\end{proofof}
To show Lemma~\ref{lem:smoothableMGF} will need the following result
by Tropp~\cite{tropp2012user},
which is a corollary of Lieb's Concavity Theorem~\cite{Lieb73}.
\begin{theorem}
\label{thm:troppjensen}
Given a random symmetric matrix $Z$,
and a fixed symmetric matrix $H$,
\[
\expt
{\trop{\exp \paren{H+Z}}}
\leq
\trop{
\exp\paren{H+\log \expt \exp(Z)}
}.
\]
\end{theorem}
We will use the following claim to control the above trace by an
inductive argument.
\begin{claim}
\label{clm:smoothableInduction}
For all $j=1,\ldots,k,$ and all $\theta$ such that $0 < \theta^{2} \le
\min \{ \frac{1}{\sigma_{1}^{2}}, \frac{5}{12\sigma_{2}^{2}}\}$, we have,
\[ \seqexpt{j}
{
\trop{\exp
\paren{
\theta^{2}
\paren{\sum_{i = j+1}^{k} 
\Omega_{i}
}
+
\theta \sum_{i=1}^{j}
   \sum_{e = 1}^{l_{i}}Z^{(i)}_{e}
  }}}
\leq 
\seqexpt{j-1}
{
\trop{\exp
\paren{
\theta^{2}
\paren{\sum_{i = j}^{k} 
\Omega_{i}
}
+
\theta \sum_{i=1}^{j-1}
   \sum_{e = 1}^{l_{i}}Z^{(i)}_{e}
  }} }
.
\]
\end{claim}
Before proving this claim, we see that it immediately implies
Lemma~\ref{lem:smoothableMGF}.
\begin{proofof}{of Lemma~\ref{lem:smoothableMGF}}
  We chain the inequalities given by
  Claim~\ref{clm:smoothableInduction} for $j  = k,k-1,\ldots,1$ to
  obtain,
\begin{align*}
\seqexpt{k}
{  \trop{\exp \paren{\theta Z}} }
 = \seqexpt{k}
{  \trop{\exp \paren{\theta \sum_{i=1}^{k}\sum_{e = 1}^{l_{i}}Z^{(i)}_{e} }}}  
\le
\trop{\exp
\paren{
\theta^{2}
\sum_{i = 1}^{k} 
\Omega_{i}
  }}
\leq
d\exp\paren{\theta^2 \sigma_{3}^{2} },
\end{align*}
where the last inequality follows from
$\trop \exp \paren{A} \le d \exp \paren{\norm{A}}$ for symmetric $A$
and
$\norm{\sum_{i = 1}^{k} 
\Omega_{i}}
\le \sigma_{3}^{2}.$
\end{proofof}
We will also need the next two lemmas, which essentially appear in
the work of Tropp~\cite{tropp2012user}.
For completeness, we also prove these lemmas in Appendix~\ref{sec:mgfBounds}.
\begin{lemma}
\label{lem:mgfby2ndmoment}
  Suppose $Z$ is a random matrix s.t. $\expt Z \preceq 0$,
  $\norm{Z} \leq 1$, then,
  $
  \log \expt \exp \paren{ Z} \preceq \frac{4}{5} \expt Z^{2}.
  $
\end{lemma}
\begin{lemma}
\label{lem:mgfby1stmoment}
  Suppose $C$ is a PSD random matrix s.t.
  $\norm{C} \leq 1/3$, then,
  $
  \log \expt \exp \paren{C} \preceq  \frac{6}{5} \expt C .
  $
\end{lemma}
We also need the following well-known fact (see for example
\cite{tropp2012user}):
\begin{fact}
  \label{fact:trexp-monotone}
  Given symmetric matrices $A,B$ s.t. $A \preceq B$,
  $\trop{\exp\paren{A}} \preceq \trop{\exp\paren{B}}$.
\end{fact}

\begin{lemma}
\label{lem:pushedges}
For all $\theta$ such that $0 < \theta^{2} \le \frac{1}{\sigma^{2}_{1}},$ all $j=1,\ldots,k$ and for
all symmetric $\widetilde{H}$ that are fixed given
$(r_{1}, R_{1}), \ldots, (r_{j-1},R_{j-1}),$
\[ 
\seqexpt{j}
{  \trop{\exp\paren{
\widetilde{H}
+
\theta 
  \sum_{e} Z^{(j)}_{e}} } }
\leq 
\seqexpt{j-1}
\expt_{r_{j}} 
{  \trop{\exp\paren{
\widetilde{H}
+
\frac{4}{5} \theta^{2} \expt_{R^{(j)}_{e}}  \sum_{e} (Z_e^{(j)})^{2}
} } }.\]
\end{lemma}
\begin{proof}
We have
\begin{align*}
\seqexpt{j}{
\trop{\exp\paren{
\widetilde{H}
+
\theta
  \sum_{e} Z^{(j)}_{e}
  }} }
& = \seqexpt{j-1}
\expt_{r_j}
\expt_{R^{(j)}_{1}}
\ldots
\expt_{R^{(j)}_{l_j}}
{
\trop{\exp\paren{
\widetilde{H}
+
\theta \sum_{e = 1}^{l_{j}}Z^{(j)}_{e}
  }} } \\
& \qquad\textrm{Using
  Theorem~\ref{thm:troppjensen} with $\textstyle H = \widetilde{H}
+ 
\theta \sum_{e = 1}^{l_j-1} Z^{(j)}_{e}$} \\
& \le \seqexpt{j-1}
\expt_{r_j}
\expt_{R^{(j)}_{1}}
\ldots
\expt_{R^{(j)}_{l_{j}-1}}
{
\trop{\exp\paren{
\widetilde{H}
+ 
\theta \sum_{e = 1}^{l_j-1} Z^{(j)}_{e}
+ 
\log \expt_{R^{(j)}_{l_j}} \exp \paren{\theta Z^{(j)}_{l_j} }
  }} } \\
& \qquad \vdots \qquad \qquad \qquad \qquad\qquad \qquad\textrm{(Using
  Theorem~\ref{thm:troppjensen} $l_j-1$ times)} \\
& \le \seqexpt{j-1}
\expt_{r_j}
{
\trop{\exp\paren{
\widetilde{H}
+
\sum_{e = 1}^{l_j} 
\log \expt_{R^{(j)}_{e}} \exp \paren{\theta Z^{(j)}_{e} }
  }} }
\end{align*}
For all $j,e,$ we have $\expt_{R^{(j)}_{e}} \theta Z_e^{j} \preceq 0,$ and
$\norm{\theta Z_{e}^{(j)}} \le 1$ (since
$\norm{Z_{e}^{(j)}} \le \sigma_{1}$ and $\theta \le \frac{1}{\sigma_{1}}),$ thus
Lemma~\ref{lem:mgfby2ndmoment} gives 
\[  \sum_{e}\log \expt_{R^{(j)}_{e}} \exp \paren{\theta Z^{(j)}_{e} } \preceq
\frac{4}{5} \theta^{2}  \sum_{e}\expt_{R^{(j)}_{e}} (Z_{e}^{(j)})^{2}.\]
Now, using Fact~\ref{fact:trexp-monotone} which states that
$\trop{\exp\paren{\cdot}}$
is monotone increasing
with respect to the PSD order, we obtain the lemma.
\end{proof}

\begin{proofof}{of Claim~\ref{clm:smoothableInduction}}
Since $0 < \theta^{2} \le \frac{1}{\sigma_{1}^{2}},$
using Lemma~\ref{lem:pushedges} with $
  \widetilde{H} = \theta^{2} \paren{\sum_{i = j+1}^{k} \Omega_{i} }
+
\theta \sum_{i=1}^{j-1}
  \sum_{e} Z^{(i)}_{e},$
 we obtain,
\begin{align*}
\seqexpt{j}
{
\trop{\exp\paren{
\theta^{2}
\paren{\sum_{i = j+1}^{k} \Omega_{i} }
+
\theta \sum_{i=1}^{j}
  \sum_{e}  Z^{(i)}_{e}
  }} } 
\leq \seqexpt{j-1}
\expt_{r_{j}} 
{  \trop{\exp\paren{
\theta^{2}
\paren{\sum_{i = j+1}^{k} \Omega_{i} }
+
\theta \sum_{i=1}^{j-1}
  \sum_{e}  Z^{(i)}_{e}
+ \frac{4}{5} \theta^{2}  \sum_{e} \expt_{R^{(j)}_{e}}   (Z_e^{(j)})^{2}} } } \\
\leq \seqexpt{j-1}
{  \trop{\exp\paren{
\theta^{2}
\paren{\sum_{i = j+1}^{k} \Omega_{i} }
+
\theta \sum_{i=1}^{j-1}
  \sum_{e}  Z^{(i)}_{e}
+ \log \expt_{r_{j}} \exp\paren{\frac{4}{5} \theta^{2} \sum_{e} \expt_{R^{(j)}_{e}}   (Z_e^{(j)})^{2} }
} } },
\end{align*}
where the last inequality uses Theorem~\ref{thm:troppjensen} with $H =
\theta^{2}
\paren{\sum_{i = j+1}^{k} \Omega_{i} }
+
\theta \sum_{i=1}^{j-1}
  \sum_{e}  Z^{(i)}_{e}$.
Now, since $0 \preceq \sum_{e}\expt_{R^{(j)}_{e}} (Z_e^{(j)})^{2},$ and
$\norm{\frac{4}{5} \theta^{2}  \sum_{e}\expt_{R^{(j)}_{e}} (Z_e^{(j)})^{2}
} \le \frac{4}{5} \theta^{2} \sigma_{2}^{2} \le \frac{1}{3},$ thus
Lemma~\ref{lem:mgfby1stmoment} gives
\[ \log \expt_{r_{j}} \exp\paren{\frac{4}{5} \theta^{2}
\sum_{e} 
\expt_{R^{(j)}_{e}}
(Z_e^{(j)})^{2} } 
\preceq
\frac{6}{5}\frac{4}{5} \theta^{2} \expt_{r_{j}}
\sum_{e}
\expt_{R^{(j)}_{e}}
 (Z_e^{(j)})^{2}
\preceq
\theta^{2} \expt_{r_{j}}
\sum_{e}
\expt_{R^{(j)}_{e}} 
(Z_e^{(j)})^{2}
\preceq
\theta^{2} \Omega_{j}.
\]
Now, using Fact~\ref{fact:trexp-monotone}, namely that $\trop{\exp\paren{\cdot}}$ is
monotone increasing with respect to the PSD order, we obtain the lemma.
\end{proofof}

\subsection{Truncating {\SEQSUM}s}
To prove this theorem, we use the next lemma, which we will prove
later in this section:
\begin{lemma}
\label{lem:goodUntilByIndicators}
\[
\goodUntil_{k+1} =
\left[  \truncZ \preceq \eps
  I \right]
.
\]
\end{lemma}
\newcommand{\Zsum}[2]{\sum_{#2}  Z^{(#1)}_{#2}}
\begin{proofof}{of Theorem~\ref{thm:truncSum}}

Note that $\goodUntil_{k+1}$ implies $Z \preceq \eps I$.
Thus $Z \not\preceq \eps I$ implies $\neg \goodUntil_{k+1}$,
or equivalently ${[Z \not\preceq \eps I] \subseteq
  \neg\goodUntil_{k+1}}$.  Using the above lemma, it immediately
implies the first part of the claim.
Also note that $\goodUntil_{k+1}$ implies $\goodUntil_{i}$ for all $i \leq
k$.
Thus, if  $\goodUntil_{k+1}$ occurs then $\ind{\goodUntil_{i}} = 1$
for all $i \leq k$.
\begin{align*}
\Pr[ -\eps I \not\preceq Z
\OR
Z \not\preceq \eps I 
]
&\leq
\Pr\left[ -\eps I \not\preceq \sum_{i = 1}^{k} \Zsum{i}{e}
\OR
\neg\goodUntil_{k+1}
\right]
\\
&= 
\Pr\left[
 \left( -\eps I \not\preceq \sum_{i = 1}^{k}\left(
  \ind{\goodUntil_{i}} \Zsum{i}{e} \right)
 \AND \goodUntil_{k+1} 
 \right)
\OR
\neg \goodUntil_{k+1} 
\right]
\\
&= 
\Pr\left[
-\eps I \not\preceq \sum_{i = 1}^{k} \left( \ind{\goodUntil_{i}}
  \Zsum{i}{e} \right)
\OR
\neg \goodUntil_{k+1} 
\right]
=
\Pr[ -\eps I \not\preceq \truncZ
\OR
 \truncZ \not\preceq \eps I ]
.
\end{align*}
In the last line, we used Lemma~\ref{lem:goodUntilByIndicators}.

$\truncZ$ is a {\seqsum}, because conditional on $(i-1)$ and $r_{i}$,
the indicator $\ind{\goodUntil_{i}}$ is fixed at either $0$ or $1$, and so in both
cases, the variables $\ind{\goodUntil_{i}} Z^{(i)}_{e}$ are
independent and zero mean.
\end{proofof}
\begin{proofof}{Lemma~\ref{lem:goodUntilByIndicators}}
We start by showing that $\goodUntil_{k+1}$
implies the event
$\left[ 
\sum_{i = 1}^{k} \left( \ind{\goodUntil_{i}}
  \Zsum{i}{e}
   \right)
\preceq \eps
  I \right]$.
Suppose $\goodUntil_{k+1}$ occurs.
Then $\goodUntil_{i}$ occurs for
all $i\leq k+1$,
so $ \ind{\goodUntil_{i}} = 1$ for all $i\leq k+1$.
But then, ${\sum_{i = 1}^{k} \ind{\goodUntil_{i}}  \Zsum{i}{e}  = \sum_{i =
  1}^{k}  \Zsum{i}{e} }$ and
  $\sum_{i =
  1}^{k}  \Zsum{i}{e} \preceq \eps
  I $ follows from $\goodUntil_{k+1}$, so $\sum_{i = 1}^{k}
  \ind{\goodUntil_{i}}  \Zsum{i}{e} \preceq \eps I$.
This proves that $\goodUntil_{k+1}$
implies the event $\left[  \sum_{i = 1}^{k} \ind{\goodUntil_{i}}  \Zsum{i}{e} \preceq \eps
  I \right]$.

Next we show that $\left[  \sum_{i = 1}^{k} \ind{\goodUntil_{i}}  \Zsum{i}{e} \preceq \eps
  I \right]$ implies the event $\goodUntil_{k+1}$.
It suffices to show the contrapositive: $\neg
\goodUntil_{k+1}$
implies $\neg \left[  \sum_{i = 1}^{k} \ind{\goodUntil_{i}}  \Zsum{i}{e} \preceq \eps
  I \right]$.
Suppose $\neg
\goodUntil_{k+1}$ occurs.
This means there must exist $j \leq k$ s.t. 
  $\sum_{i = 1}^{j} \Zsum{i}{e} \not\preceq \eps I$.
Let $j^{*}$ denote the least such $j$.
Observe that $ \ind{\goodUntil_{i}} = 1$ for all $i\leq j^{*}$, and $
\ind{\goodUntil_{i}} = 0$ for all $i > j^{*}$. 
Thus
\[
\sum_{i = 1}^{k} \ind{\goodUntil_{i}}  \Zsum{i}{e}
=
\sum_{i = 1}^{j^{*}} \Zsum{i}{e}
\not\preceq \eps
  I
.
\]
So $\neg \left[  \sum_{i = 1}^{k} \ind{\goodUntil_{i}} \Zsum{i}{e} \preceq \eps
  I \right]$ occurs.
\end{proofof}



\section*{Acknowledgements}
We thank Daniel Spielman for suggesting this project and for helpful
comments and discussions.

\bibliographystyle{alpha}
\bibliography{papers}

\appendix
\section{Conditions for Bounding Matrix Moment Generating Functions}
\label{sec:mgfBounds}
\begin{proofof}{of Lemma~\ref{lem:mgfby2ndmoment}}
We define $f(z) = \frac{e^{z} - z - 1}{z^{2}}$. Note that $f(1) \leq 4/5$.
 The function $f$ is positive and increasing in $z$ for all
real $z$.
This means for every symmetric matrix $A$,
$f(A) \preceq f(\norm{A}) I $, and so for any symmetric matrix $B$,
$Bf(A)B \preceq f(\norm{A}) B^{2} $. Thus
\begin{align*}
\expt
\exp \paren{Z}
=
\expt \sqbr{ I + Z + Z f(\norm{Z}) Z  }
 \preceq
I + 0 + f(\norm{Z}) \expt Z^{2}
\end{align*}
The lemma now follows from using the fact that $\log$ is operator
monotone (increasing), that $\log(1+z) \leq z$ for all real
$z > 0,$ and $\expt[Z^{2}] \succeq 0 $.
\end{proofof}

\begin{proofof}{of Lemma~\ref{lem:mgfby1stmoment}}
We define $g(z) = \frac{e^{z} - 1}{z}$.
 The function $g$ is positive and increasing in $z$ for all
real $z \geq 0$.
This means for every symmetric matrix $A$,
$g(A) \preceq g(\norm{A}) I $, and so for any symmetric matrix $B$,
$Bg(A)B \preceq g(\norm{A}) B^{2} $.
Also $g(1/3) \leq 6/5 $.
Thus
\begin{align*}
\expt
\exp \paren{C}
=
\expt \sqbr{ I + C^{1/2}g(C)C^{1/2} }
\preceq
I + g(1/3) \expt C
\preceq
I + \frac{6}{5} \expt C
.
\end{align*}
The lemma now follows using the fact that $\log$ is operator monotone (increasing),
$\log(1+z) \leq z$ for all real $z > 0,$ and $C \succeq 0.$ 
\end{proofof}


\section{Obtaining Concentration of Running Time}
\label{sec:timeConcentration}
As indicated in Remark~\ref{rem:timeConc}, we can obtain a version of
Theorem~\ref{thm:sparsecholesky} that provides running time guarantees
with high probability instead of in expectation,
by making a slight change to the $\cholesky$ algorithm.
In this appendix, we briefly sketch how to prove this.
We refer to the modified algorithm as $\concChol$.
The algorithm is requires only two small modifications:
Firstly, instead of initially choosing a permutation at random,
we choose the $k^{\text{th}}$ vertex to eliminate by sampling it
uniformly at random amongst the remaining vertices that have degree at
most twice the average multi-edge degree in $\samp{S}^{(k-1)}$.
We can do this by keeping track of all vertex degrees, and sampling a
remaining vertex at random, and resampling if the degree is too high,
until we get a low degree vertex.
Secondly, to make up for the slight reduction in the randomization of
the choice of vertex, we double the value of $\rho$ used in Line~\ref{alg:solver:initSplit}.

We get the following result:
\begin{theorem}
  \label{thm:concSparseCholesky}
  Given a connected undirected multi-graph
  $G =(V,E)$, with positive edges weights 
  $w : E \to \rea_{+}$, and associated Laplacian $L$,
  and scalars $\delta >1$, $0<\eps\leq1/2$,
  the algorithm $\concChol(L,\eps,\delta)$
  returns a sparse approximate Cholesky decomposition
  $(P,\mathcal{L},\mathcal{D})$ s.t. 
  with probability at least $1-2/n^{\delta}$,
  \begin{align}
    \label{eq:lapErrorBounds}
   (1-\eps) L
   \preceq
   P \mathcal{L} \mathcal{D} \mathcal{L}^{\top} P^{\top} 
   \preceq
   (1+\eps) L   .
  \end{align}
  The number of non-zero entries in $\mathcal{L}$ is $O(\frac{\delta^{2}}{\eps^{2}}
  m\log^{3} n)$.
  With high probability the algorithm runs in time $O(\frac{\delta^{2} }{\eps^{2}}
  m\log^{3} n)$.
\end{theorem}
\paragraph{Proof Sketch.}
The new procedure for choosing a random vertex will 
select a vertex uniformly among the remaining vertices with degree at most twice
the average degree in $\samp{S}^{(k-1)}$. An application of Markov's theorem tells us that at least
half the vertices in $\samp{S}^{(k-1)}$ will satisfy this.
The only change in our matrix concentration analysis that results from
this is that Lemma~\ref{lem:csampmoments}
Part~\ref{lem:csampmoments:cexpt} will lose a factor 2 and becomes:
$ \expt_{\pi(k)} \ind{\goodUntil_{k}} \lnorm{C_{\pi(k)}(\samp{S}^{(k-1)})} 
  \preceq
   \frac{4(1+\eps)}
    {n+1-k} 
    I
.
$
This means that when applying the
Theorem~\ref{thm:smoothableTailProb},
our bound on $\sigma_{3}^{2}$ will be worse by a factor 2.
Doubling $\rho$ will suffice to obtain the same concentration bound as
in Theorem~\ref{thm:sparsecholesky}.
Next, the running time spent on calls on $\csamp$ in the
$k^{\text{th}}$ step will now be
deterministically bounded by $O(\rho m / (n+1-k))$.
Finally, since we pick each vertex using rejection sampling,
we have to bound the time spent picking each vertex.
Each resampling will take $O(1)$ time.
The number of samples required to pick one $v$ will be distributed as
a geometric variable with success probability at least $1/2$.
Thus, the total number of vertex resamplings is distributed as a sum of
$n-1$ independent geometric random variables with success probability $1/2$.
The sum of these variables will be bounded by $10 n$ w.h.p., so the total
time of the algorithm is bounded by $O(\frac{\delta^{2}}{\eps^{2}}
  m\log^{3} n)$ w.h.p.


\newcommand{\bddChol}{\textsc{BDDSparseCholesky}}
\newcommand{\bdd}{bDD}
\newcommand{\sparsify}{\textsc{Sparsify}}
\section{Sparse Cholesky Factorization for {\bdd} Matrices}
\label{sec:bdd}
In this appendix, we sketch a version of our approximate Cholesky factorization algorithm,  that is also
applicable to BDD matrices, which include the class of Connection
Laplacian matrices. 
We call this algorithm $\bddChol$.
It follows closely the algorithmic structure used in
\cite{KyngLPSS16}.
Like~\cite{KyngLPSS16},
we need the input matrix to be non-singular, which we can achieve using the
approach described in Claims 2.4 and 2.5 of~\cite{KyngLPSS16}.

Our algorithm  replace their expander-based Schur complement
approximation routine with our simple one-by-one vertex elimination,
while still using their recursive subsampling based framework for
estimation of leverage scores.
The constants in this appendix are not optimized.

We study {\bdd} matrices as defined in~\cite{KyngLPSS16},
with $r \times r$ blocks, where $r$ is a constant (see their Section
1.1).
The class of {\bdd} matrices is Hermitian, rather than symmetric, 
but our notion of spectral approximation and
our matrix concentration results extend immediately to Hermitian
matrices.
Throughout this section, we use $(\cdot)^{\dag}$ to conjugate
transpose.
Our algorithm will still compute a Cholesky composition,
except we do not factor the individual $r \times r$
block matrices.
We will sketch a proof of the following result:
\begin{theorem}
\label{thm:intro:main}
The algorithm ${\bddChol},$
given an $nr \times nr$ {\bdd} matrix $L$
with $m$ non-zero blocks entries, runs in time $O(m \log^{3} n+n \log^{5} n)$ whp.
and computes a permutation $\pi,$ a lower triangular matrix $\mathcal{L}$
with $O(m \log^{2} n+n\log^{4} n)$ non-zero entries, and a diagonal matrix
$\mathcal{D}$ such that with probability $1-\frac{1}{\poly(n)},$
we have
\[\nfrac{1}{2} \cdot L \preceq 
Z
 \preceq \nfrac{3}{2} \cdot L,\]
where $Z = P_{\pi} \mathcal{L} \mathcal{D} \mathcal{L}^{\dag}
P_{\pi}^{\dag},$ \emph{i.e.}, $Z$ has a sparse Cholesky factorization. 
\end{theorem}
We choose a fixed $\eps = 1/2$, but the algorithms can be adapted to
produce approximate Cholesky decompositions with $\eps \leq 1/2$
spectral approximation and running time dependence $\eps^{-2}$.

For BDD matrices, we do not know a result analogous to
the fact that effective resistance is a distance in Laplacians (see
Lemma~\ref{lem:reffDist}).
Instead, we use a result that is weaker by a factor $2$:
Given two {\bdd} multi-edge matrices $e$, $e'$
$w(e) B_{u,v}B_{u,v}^{\dag}$ and
$w(e') B_{v,z} B_{v,z}^{\dag}$ that are incident on vertex blocks $v,u$ and
$v,z$ respectively, if we eliminate vertex block $v$,
this creates a multi-edge $e''$ with BDD matrix $w(e'') B_{u,z} B_{u,z}^{\dag}$ 
satisfying 

\begin{equation}
\norm{\lnorm{B_{u,z} B_{u,z}^{\dag}}}
\leq
2
\left(
\norm{\lnorm{B_{v,u} B_{v,u}^{\dag}}}
+
\norm{\lnorm{B_{v,z} B_{v,z}^{\dag}}}
\right)
\label{eq:bddReff}
\end{equation}

We will sketch how to modify the $\concChol$ algorithm to solve BDD matrices.
We call this $\bddChol$.
This algorithm is similar to $\cholesky$ and $\concChol$, except that
the number of multi-edges in the approximate factorization will be
slowly increasing with each elimination and to counter this we will
need to occasionally sparsify the matrices we produce.
First we will assume an oracle procedure $\sparsify$, which we will
later see how to construct using a boot-strapping approach that
recursively calls $\bddChol$ on smaller matrices.

\paragraph{Solving {\bdd} matrices using a sparsification oracle.}
We assume the existence of a procedure $\sparsify$.
Given a $nr \times nr$  {\bdd} matrix $S$, s.t. $S \preceq 2L$,
and $S$ is $1/\rho$-bounded
w.r.t. $L$, 
 $\sparsify(S,\rho)$ returns $\rho^{2} \cdot 2\cdot10^{5} nr$ IID distributed
samples $Y_{e}$ of multi-edges s.t. $\expt \sum_{e} Y_{e} = S$,
and each sample is $1/\rho$-bounded w.r.t. $L$.

$\bddChol(S,\rho)$ should be identical to $\concChol$,
except
\begin{enumerate}
\item The sampling rate $\rho$ is an explicit parameter to
$\bddChol$.
\item $\bddChol$ should not split the initial input multi-edge into
  $\rho$ smaller copies.
This will be important because we use $\bddChol$ recursively,
and we will only split edges at the top level.
\item We adapt the $\csamp$ routine to sample from a {\bdd}
  elimination clique,
and we produce $2 \mult{S}{v}$ samples, and scale all samples by a
factor $1/2$.
\item  After eliminating $\frac{9}{10} n$ vertices,
it calls $\sparsify(\samp{S}^{(9n/10)},\rho)$  to produce a sampled matrix
$S'$ of dimension ${nr/10 \times nr/10}$, and then calls $\bddChol(S',\rho)$
to recursively produce an approximate Cholesky decomposition of $S'$.
\item To compute an approximate Cholesky decomposition of $L$,
we set $\rho = \ceil{10^{3} \log^{2} n}$.\\
Form $S^{(0)}$ by splitting each edge of $L$ into $\rho$ copies with $1/\rho$ of their
initial weight. \\
Call $\bddChol(S^{(0)},\rho)$.
\end{enumerate}
The clique sampling routine for {\bdd} matrices uses more conservative
sampling than $\csamp$ ,
because we use the weaker Equation~\eqref{eq:bddReff}.
The sparsification step then becomes necessary because our clique
sampling routine now causes the total number of number of multi-edges
to increase with each elimination.
However, the increase will not exceed a factor $(1+8/n)$,
so after $\frac{9}{10} n$ eliminations, the total number of
multi-edges has not grown by more than $2 \cdot 10^{4}$.

We can use a {\truncSum} to analyze the entire approximate Cholesky
decomposition produced by $\bddChol$ and its recursive calls using
Theorems~\ref{thm:smoothableTailProb} and~\ref{thm:truncSum} (these
theorems extend to Hermitian matrices immediately).
The calls to $\sparsify$ will cause our bound on the martingale
variance $\sigma_{3}$ to grow
larger, but only by a constant factor.
By increasing $\rho$ by an appropriate constant, we still obtain concentration.

On the other hand, the calls to $\sparsify$ will ensure that the time
spent in the recursive calls to $\bddChol$ is only a constant fraction
of the time spent in the initial call, assuming the total number of
multi-edges in $S$ exceeds $\rho^{2} \cdot 2 \cdot 10^{6} nr$ (if not, we can
always split edges to achieve this).
This corresponds to assuming that before the initial edge splitting at
the start of the algorithm, we have at least $\rho \cdot 2 \cdot
10^{6} nr$ edges.

This means the total time to compute the approximate Cholesky
decomposition of $L$ using $\bddChol$ will only be $O(\rho (m + \rho n))$,
excluding calls to $\sparsify$.
The decomposition will have $O(m\log^{2} n+n\log^{4} n)$ non-zeros, and
its approximate inverse can be applied in $O(m \log^{2} n +n\log^{4} n)$ time.

We now sketch briefly how to implement the $\sparsify$ routine.
It closely resembles the $\sparsify$ routine of
\cite{KyngLPSS16} (see their Lemma H.1).

\paragraph{Implementing the sparsification routine.}
$\sparsify(\samp{S}^{(9n/10)},\rho)$ uses the subsampling-based techniques of
\cite{CohenLMMPS15}.
It is identical to the sparsification routine of \cite{KyngLPSS16},
except the recursive call to a linear solver uses $\bddChol$.
The routine first samples each multi-edge with
probability $\frac{1}{2 \cdot 10^{5} \rho}$ to produce a sparse matrix $S''$,
and then uses Johnson-Lindenstrauss-based leverage score estimation
(see \cite{CohenLMMPS15}) to compute IID samples with the desired $1/\rho$-bound
w.r.t. $L$.
The IID samples are summed to give the output matrix $S'$.
This requires approximately solving $\sqrt{\rho}$ systems of linear
equations in the sparse matrix $S''$.
To do so, $\sparsify$ first splits every edge of $S''$ into $\rho$
copies (increasing the number of multi-edges by a factor $\rho$),
then makes a single recursive call to $\bddChol(S'',\rho)$,
and then uses the resulting Cholesky decomposition $\sqrt{\rho}$ times
to compute approximate solutions to systems of linear equations.
One issue requires some care: The subsampling guarantees provided by
\cite{CohenLMMPS15} are with respect to $\samp{S}^{(9n/10)}$ and not $L$,
however, by using a {\truncSum} in our analysis, we can assume that
$\samp{S}^{(9n/10)}\preceq 2 L$.

\paragraph{Running time including sparsification.}
Finally, if we take account of time spent on calls to
$\sparsify$ and its recursive calls to $\bddChol(S'',\rho)$,
we get a time recursion for $\bddChol$ of 
\[
T(m) \leq 10^{7} \rho m+ 100 \rho^{0.5} \rho m
+2T(m/10),
\]
(assuming initially for $L$
that $m \geq \rho 2 \cdot 10^{6} n$),
which can be solved to give a running time bound of $O(m \rho^{1.5} + n \rho^{2.5} )
= O(m \log^{3} n+ n \log^{5} n)$.

Fixing $\rho$ at the top level ensures that a union bound
across all recursive calls in $\bddChol$ will give that the
approximate Cholesky decomposition obtains a $1/2$ factor spectral
approximation with high probability.

\begin{remark}
\label{rem:slightlyFaster}
By applying the sparsification approach described in this section
once, we can compute an approximate Cholesky decomposition of
Laplacian matrices in time 
$O(m \log^2 n \log\log n)$ time w.h.p.

We run $\concChol$ with a modification:
after eliminating all but $n/log^{100}n$
vertices, we do sparsification on the remaining graph with $\rho m$ multi-edges and $n/log^{100}n$
vertices. The call to $\concChol$ until sparsification will take $O(m
\log^2 n \log\log n)$ time w.h.p.
The sparsification is done using a modified version of the $\sparsify$
routine described above.
Instead of sampling each multi-edge with probability $\frac{1}{2 \cdot
  10^{5} \rho}$, we use a probability of $\frac{1}{\log^{8} n}$.
The recursive linear solve in $\sparsify$ can be done using
unmodified $\concChol$, and $O(\log n)$ linear system solves for
Johnson-Lindenstrauss leverage score estimation can be done using this
decomposition, all in $O(m+n)$ time.
The output graph $S'$ from $\sparsify$ can be Cholesky decomposed
using unmodified $\concChol$ as well, and this will take time $O(m+n)$.
In total, we get a running time and number of non-zeros bounded by
$O(m \log^2 n \log\log n)$ time w.h.p.
\end{remark}


\end{document}